\documentclass[11pt]{article}
\usepackage{times,amsmath,amssymb,theorem}%
\newfont{\fsc}{eusm10 scaled 1100}      %
\newfont{\bbb}{bbm10 scaled 1100}       %
\makeatletter
\def\section{\@startsection{section}{1}{0pt}{-3.25ex plus -1ex minus 
-.2ex}{1.5ex plus .2ex minus .3ex}{\normalfont\large\bf}}
\renewcommand\subsection{\@startsection{subsection}{2}{\z@}%
                                     {-3.25ex\@plus -1ex \@minus -.2ex}%
                                     {1.5ex \@plus .2ex}%
                                     {\normalfont\normalsize\bfseries}}
\makeatother
\theoremheaderfont{\scshape}
\theorembodyfont{\normalfont\slshape}
\theoremstyle{plain}
\newtheorem{lemma}{Lemma}
\newtheorem{theorem}{Theorem}
\newenvironment{proof}{\begin{trivlist}\item{}\normalfont\textit{Proof.}}{\hfill$\square$\end{trivlist}}
\newcommand{\defn}[1]{{\textit{\textbf{#1}}}}
\newcommand{\ie}{\emph{i.e.}}
\newcommand{\Ie}{\emph{I.e.}}
\newcommand{\nth}{^{\text{th}}}
\newcommand{\leftNoCarry}[1]{\circ\;#1\;\;\;}
\newcommand{\rightNoCarry}[1]{\;\;\;#1\;\circ}
\newcommand{\leftCarry}[1]{\bullet\;#1\;\;\;}
\newcommand{\rightCarry}[1]{\;\;\;#1\;\bullet}
\newcommand{\varone}{\underline{1}}
\newcommand{\varzero}{\underline{0}}
\newcommand{\dunno}{?}
\newcommand{\wolframcarry}{\begin{picture}(6,0)\put(3,2){\makebox(0,0){$\bullet$}}\put(3.09,0){\makebox(0,0){\tiny$\blacktriangledown$}}\end{picture}}
\newcommand{\wolframnocarry}{\begin{picture}(6,0)\put(3,2){\makebox(0,0){$\bullet$}}\put(3.09,4.4){\makebox(0,0){\tiny$\blacktriangle$}}\end{picture}}
\newcommand{\myitem}[1]{\item[\normalfont(#1)]}
\newcommand{\repleft}[1]{\raisebox{.3ex}{\footnotesize$\Leftarrow$}{#1}}
\newcommand{\repright}[1]{{#1}\raisebox{.3ex}{\footnotesize$\Rightarrow$}}
\newcommand{\repleftcell}[1]{\cell{\makebox[0ex][r]{\raisebox{.3ex}{\footnotesize$\Leftarrow$}}{#1}}}
\newcommand{\reprightcell}[1]{\cell{{#1}\makebox[0ex][l]{\raisebox{.3ex}{\footnotesize$\Rightarrow$}}}}
\newcommand{\cellw}{2.2ex}
\newcommand{\cell}[1]{\makebox[\cellw]{$#1$}}
\newcommand{\cells}{\hspace*{-1ex}\begin{array}{c}}
\newcommand{\cellsleft}{\hspace*{-1ex}\begin{array}{@{\hspace*{-30ex}}l@{\hspace*{-30ex}}}}
\newcommand{\cellsend}{\end{array}\hspace*{-1ex}}
\newcommand{\nextrow}{\\[.5ex]}

\newcommand{\nocarry}[1]{\overset{\circ}{\raisebox{0ex}[1.4ex]{$#1$}}}

\newcommand{\tapedata}{I}%
\newcommand{\future}[1]{\mbox{\tiny$\swarrow$\hspace*{-.2ex}} #1\mbox{\tiny\hspace*{-.2ex}$\searrow$}}
\newcommand{\codedtaginput}{I}

\newcommand{\Z}{\mathbb{Z}}
\newcommand{\N}{\mbox{\bbb N}}
\newcommand{\leftseed}[1]{{}^{\mbox{\tiny$\ulcorner$}\hspace*{-.8ex}}#1}
\newcommand{\rightseed}[1]{#1^{\hspace*{-.2ex}\mbox{\tiny$\urcorner$}}}
\newcommand{\leftstem}[1]{{}_{\mbox{\tiny$\llcorner$}\hspace*{-.55ex}}#1}
\newcommand{\rightstem}[1]{#1_{\hspace*{-.4ex}\mbox{\tiny$\lrcorner$}}}
\newcommand{\one}{1}
\newcommand{\zero}{0}
\newcommand{\tinysym}[1]{\raisebox{.2ex}{\tiny #1}}
\newcommand{\tinyzero}{\tinysym 0}
\newcommand{\tinyone}{\tinysym 1}
\newcommand{\tinyvarzero}{\tinysym\varzero}
\newcommand{\tinyvarone}{\tinysym\varone}
\newcommand{\tinydunno}{\tinysym\dunno}
\newcommand{\fc}[1]{\makebox[1.8ex]{$#1$}}
\newcommand{\tc}[1]{\makebox[1.2ex]{$#1$}}
\newcommand{\remainder}{\%}%
\newcommand{\matr}[1]{\widehat{#1}}
\newcommand{\leftwraptraj}[1]{\overleftarrow {T_{#1}}}
\newcommand{\rightseedtraj}[1]{\overrightarrow {T_{#1}}}
\newcommand{\concat}{\textsf{;}}
\newcommand{\lefttrajcomp}[1]{\overleftarrow{#1}}
\newcommand{\rightwraptraj}[1]{\overrightarrow {T_{#1}}}

\newcommand{\righttrajcomp}[1]{\overrightarrow{#1}}
\newcommand{\rev}[1]{{#1}^{\mathsf{rev}}}
\newcommand{\red}[1]{|#1|}
\newcommand{\emulation}[1]{\widehat{#1}}

\title{\vspace*{-5.5ex}\large
\textbf{Is Wolfram and Cook's (2,5) Turing machine \emph{really} universal?}
\author{\\[-3ex]
Dominic J.\ D.\ Hughes
\thanks{Visiting Scholar, Computer Science Department, 353 Serra Mall, Stanford University, Stanford CA 94305, USA.}
\\[.2ex]
\normalsize \sc Stanford University\\
\normalsize November 5, 2007}\date{}}

\begin{document}
\maketitle\vspace{-2ex}
\begin{center}\bf Abstract\vspace{-.8ex}\end{center}
\begin{quotation}\small

Wolfram \cite[p.\,707]{W} and Cook \cite[p.\,3]{C} claim to prove that
a (2,5) Turing machine (2 states, 5 symbols) is universal, via a
universal cellular automaton known as \textsl{Rule 110}.
The first part of this paper points out a critical gap in their
argument.
The second part bridges the gap, thereby giving what appears to be the
first proof of universality.

\end{quotation}

\section{The claim}

\thispagestyle{empty}
Wolfram \cite[p.\,707]{W} and Cook \cite[p.\,3]{C} claim to prove that
the Turing Machine $M$ 
with the following table is universal:
\begin{center}\label{table}\vspace{-2.5ex}\hspace*{-1ex}\small%
\begin{tabular}{r|cc}
 \multicolumn{3}{c}{\hspace{3ex}State} \\
      & $\circ$  &  $\bullet$     \\ \cline{1-3}
 $0$       & \rule{0ex}{2.8ex}$\leftNoCarry {\varzero}$  &  $\leftNoCarry \varone$    \\
 $1$       & \rule{0ex}{2.5ex}$\leftCarry {\varone}$  &  $\leftCarry \dunno$    \\ 
 \begin{picture}(0,0)\put(-42,-5){\shortstack{Input\\symbol}}\end{picture}$\varzero$  
           &  \rule{0ex}{2.5ex}$\rightNoCarry 0$  &  $\rightNoCarry 0$    \\
 $\varone$ &  \rule{0ex}{2.5ex}$\rightCarry 1$  &  $\rightCarry 1$    \\
 $\dunno$  &  \rule{0ex}{2.5ex}$\rightNoCarry 0$  &  $\rightNoCarry 1$    
\end{tabular}\end{center}
Here table entry ``$\:\bullet\: \dunno\:$'' means ``write $\dunno$ and
move left into state $\bullet$'', entry ``$\:1\,\circ\:$'' means ``write
$1$ and move right into state $\circ$'', \textit{etc}.\footnote{Cook
\cite{C} writes $0^2$ for $\varzero$, $1^2$ for $\dunno$ and
$\neq$ for $\varone$.
Wolfram on p.\,707 \cite{W} uses shades of grey (white, light grey,
medium grey, dark grey, black for $0,\varzero,\dunno,\varone,1$,
resp.), enumerated $0\ldots 4$ from light to dark on p.\,1119.
The states $\bullet\,\circ$ are $S_E\,S_O$ (resp.) for Cook, and
$\wolframcarry\,\wolframnocarry$ (resp.) for Wolfram.
The reason for our notational choice becomes clear with the example in
Section~\ref{plug}.}
For ease of reference, we collect together the passages from these
works which together constitute the claimed proof of universality:
\begin{enumerate}
\item \textnormal{\bf Wolfram, p.\,707.}\sl{}
\ldots 
by using the universality of rule 110 it turns out to be possible to
come up with the vastly simpler Turing machine shown below---with just
2 states and 5 possible colors.
\item \textnormal{\bf Wolfram, p.\,707, first caption.}\sl{}
The rule for the simplest Turing machine currently known to be
universal, based on discoveries in this book.  The machine has 2
states and 5 possible colors.
\item \textnormal{\bf Wolfram, p.\,707, second caption.}\sl{}
An example of how the Turing machine above manages to emulate rule 110.
\item \textnormal{\bf Wolfram, p.\,708.}\sl{}
As the picture at the bottom of the previous page illustrates, this
Turing machine emulates rule 110 in a quite straightforward way: its
head moves systematically backwards and forwards, at each complete
sweep updating all cells according to a single step of rule 110
evolution.  And knowing from earlier in this chapter that rule 110 is
universal, it then follows that the 2-state 5-color Turing machine
must also be universal.
\item \textnormal{\bf Wolfram, p.\,1119 (note to p.\,707).}\sl{}
\textit{Rule 110 Turing machines.} Given an initial condition for rule 110, 
the initial condition for the Turing machine shown here is obtained as
\textit{Prepend[list, 0]} with $\varzero$'s on the left and $0$'s on the right.\footnote{%
This item is not
a strictly verbatim quote: to match our notation for the Turing machine
symbols, we have substituted ``\textsl{list}'' for Wolfram's original
``\textsl{4 list}'', and ``$\varzero$'s on the left'' for ``$1$'s on
the left''.
Wolfram uses the Turing machine tape symbol `$4$' (depicted as a
solid black square) to correspond to the cellular automaton's $1$,
while we (like Cook) use the tape symbol `$1$'; where he
writes the tape symbol `$1$' (depicted as a light grey square), we write
`$\varzero$' (Cook's `$0^2$').}
\item \textnormal{\bf Cook, p.\,3.}\sl{} 
\ldots we can
construct Turing machines that are universal because they can emulate
the behavior of Rule 110.  These machines, shown in Figure~1, are far
smaller than any previously known universal Turing machines.
\item \textnormal{\bf Cook, p.\,4, caption to Figure 1.}\sl{}
Some small Turing machines which are universal due to being able to
emulate the behavior of Rule 110 by going back and forth over an ever
wider stretch of tape, each time computing one more step of Rule 110's
activity.
\end{enumerate}
This list is
the full extent of the Wolfram-Cook universality argument (aside from
Wolfram's depiction of 
an example run of the Turing
machine, p.\,707).\footnote{If there are additional details somewhere in
\cite{C,W}, they are not easy to find.  In addition, after extensive
web search, we were unable to find a universality proof.  Note: the
reader should not confuse the \emph{proof of universality of rule 110}
in Cook \cite{C}, which is laid out in full, with the \emph{claimed
proof of universality of the (2,5) Turing machine}, the full extent of
which is items 1--7 above.}
They attempt to argue as follows:
\begin{itemize}
\myitem{1} The Turing machine $M$ emulates the following cellular automaton $R\,$ (``\emph{Rule 110}''):
$$\begin{array}{|c|c|c|c|c|c|c|c|}
\hline
000&001&010&011&100&101&110&110\\
 0 & 1 & 1 & 1 & 0 & 1 & 1 & 0 \\
\hline
\end{array}\;;\vspace{-1em}$$
\myitem{2} $R$ is universal;
\myitem{3} hence $M$ is universal.
\end{itemize}

\section{The gap}\label{gap}

Unfortunately, there is a critical gap in their attempted argument.
Wolfram (item 5 above) defines the emulation of $R$ on initial
configurations 
\begin{center}
\vspace*{-1.3ex}
$\repleft0\;\;\tapedata\;\;\repright0$
\end{center}
where $\tapedata$ is a word (finite sequence of $1$'s and $0$'s), and 
$\repleft w$ (resp.\ $\repright w$) denotes the infinite repetition of
a word or symbol $w$ towards the left (resp.\ right).
However,
Wolfram and Cook demonstrated the universality of $R$
via 
initial configurations
\begin{center}
$\repleft X\;\;\tapedata\;\;\repright Y$
\end{center}
for words $X$ and $Y\mkern-5mu,$ neither constantly $0\,$.
In the former, there is an infinite stream of $0$'s either side of the
input $I$, hence a finite number of $1$'s in total. In the latter,
there are infinitely many $1$'s either side of $I$.
This breakdown in reasoning
begs the question: is
their $(2,5)$ Turing machine $M$ \emph{really}
universal?\footnote{Naive attempts to bridge the gap from $\repleft0\,
\tapedata\,\repright0$ emulations to $\repleft X\,\tapedata\,\repright Y$
emulations fail.  For example, one could
run the emulation on 
$\repleft0 \,X^n\,\tapedata\,Y^n\,\repright 0$, where $W^n$ denotes $n$
repetitions of $W$, the idea being that $X^n$ and $Y^n$ might contain
enough gliders/particles \cite{W,C} to complete the computation.
However, because of the halting problem, we can never predict how
large $n$ will need to be.  Accordingly, we could resort to repeating the
$\repleft0 \,X^n\,\tapedata\,Y^n\,\repright 0$ emulation again and
again, with progressively larger $n$, since if the target computation
on $\repleft X\,\tapedata\,\repright Y$ completes, then some $n$ will
be large enough that
$\repleft0 \,X^n\,\tapedata\,Y^n\,\repright 0$
completes in a
corresponding manner.  Alas, the \emph{deus ex machina}
(repeatedly restarting the Turing machine) destroys any possible
claim of universality.}

\section{The solution}\label{plug}

This section bridges the critical gap in the Wolfram-Cook argument, yielding what
appears to be the first proof of universality of their 2-color 5-symbol
Turing machine $M$.

Recall (item 5 above) that Wolfram defined emulation on initial configurations
\begin{center}
$\repleft0\;\;\tapedata\;\;\repright0$
\vspace*{-1.5ex}\end{center}
by running $M$ on the initial tape
\begin{center}\vspace*{-1.5ex}$
\repleft\varzero \;\; 0 \;\; \tapedata \;\; \repright 0
$\vspace*{-.8ex}\end{center}
The initial state and position of the tape head is not specified.
However, it is easy to see that it suffices to
start in state $\circ$ (Wolfram's $\wolframnocarry$) on the cell
immediately to the right of $\tapedata$:\footnote{This is consistent
with Wolfram's example depicted on p.\,707.}
\begin{center}\vspace*{-.5ex}$
\repleft\varzero \;\; 0 \;\; \tapedata \;\; \nocarry0 \;\;\repright 0
$\vspace*{-.8ex}\end{center}
$M$ simulates 
the cellular automaton $R$ thus: if $M(i)$ is the state of the tape
when the head first reaches the $i\nth$ cell right of $\tapedata$ (the
cell immediately adjacent to $I$ being counted as the $0\nth$), then
the $i\nth$ state of the cellular automaton $R$ is obtained from
$M(i)$ by replacing $\varzero$'s by $0$'s.
For example, if $\tapedata$ is $111011$
then $M(0)$ to $M(3)$ are
\begin{center}\vspace*{-3.5ex}
$\cells
\repleftcell\varzero
\cell 0
\cell 1
\cell 1
\cell 1
\cell 0
\cell 1
\cell 1
\cell {\nocarry 0}
\reprightcell 0
\nextrow
\repleftcell\varzero
\cell 0
\cell 1
\cell 1
\cell 0
\cell 1
\cell 1
\cell 1
\cell 1
\cell 0
\cell {\nocarry 0}
\reprightcell 0
\nextrow
\repleftcell \varzero
\cell 0
\cell 1
\cell 1
\cell 1
\cell 1
\cell 1
\cell 0
\cell 0
\cell 1
\cell 0
\cell 0
\cell {\nocarry 0}
\reprightcell 0
\nextrow
\repleftcell \varzero
\cell 0
\cell 1
\cell 1
\cell 0
\cell 0
\cell 0
\cell 1
\cell 0
\cell 1
\cell 1
\cell 0
\cell 0
\cell 0
\cell {\nocarry 0}
\reprightcell 0
\cellsend$
\vspace*{-.7ex}\end{center}
clearly emulating $R$.
The head sweeps repeatedly left and right, reaching one cell further
each time.
It is instructive to see the tape at the end of each leftward
sweep:
\begin{center}\vspace*{-1ex}
$\cells
\repleftcell \varzero
\cell 0
\cell 1
\cell 1
\cell 1
\cell 0
\cell 1
\cell 1
\cell {\nocarry 0}
\reprightcell 0
\nextrow
\repleftcell \varzero
\cell {\nocarry {\varzero}}
\cell \varone
\cell \dunno
\cell \dunno
\cell \varone
\cell \varone
\cell \dunno
\cell \varone
\cell \varzero
\reprightcell 0
\cell {}
\nextrow
\repleftcell \varzero
\cell 0
\cell 1
\cell 1
\cell 0
\cell 1
\cell 1
\cell 1
\cell 1
\cell 0
\cell {\nocarry 0}
\reprightcell 0
\nextrow
\repleftcell \varzero
\cell {\nocarry \varzero}
\cell \varone
\cell \dunno
\cell \varone
\cell \varone
\cell \dunno
\cell \dunno
\cell \dunno
\cell \varone
\cell \varzero
\cell \varzero
\reprightcell 0
\cell {}
\nextrow
\repleftcell \varzero
\cell 0
\cell 1
\cell 1
\cell 1
\cell 1
\cell 1
\cell 0
\cell 0
\cell 1
\cell 0
\cell 0
\cell {\nocarry 0}
\reprightcell 0
\nextrow
\repleftcell \varzero
\cell {\nocarry \varzero}
\cell \varone 
\cell \dunno
\cell \dunno
\cell \dunno
\cell \dunno
\cell \varone
\cell \varzero
\cell \varone
\cell \varone
\cell \varzero
\cell \varzero
\cell \varzero
\reprightcell 0
\cell {}
\nextrow
\repleftcell \varzero
\cell 0
\cell 1
\cell 1
\cell 0
\cell 0
\cell 0
\cell 1
\cell 0
\cell 1
\cell 1
\cell 0
\cell 0
\cell 0
\cell {\nocarry 0}
\reprightcell 0
\cellsend$
\vspace*{-1ex}\end{center}
Each leftward sweep tries to update a cell according to its right
neighbour only; if the status of the cell
cannot be determined, the head writes `$\dunno$'.  Each rightward sweep
resolves the $\dunno$'s.\footnote{This interpretation clarifies why we chose
the notation $0\,\varzero\,\dunno\,\varone\,1$ for the symbols of the Turing machine.}
The state acts as a single bit `carry', memorising the relevant
neighbour, and is set to $\bullet$
after reading a $1$ or $\varone$.

Wolfram and Cook demonstrate the universality of the cellular automaton
$R$ by emulating a universal cyclic tag system $U$.
Given an input word $J$ to $U$, they compute a word
$\codedtaginput=\codedtaginput(J)$ (a simple substitution of words for
symbols) and run $R$ with the initial configuration
\begin{center}\vspace*{-1ex}
$\repleft X\;\;\codedtaginput\;\;\repright Y$
\vspace*{-1ex}\end{center}\nopagebreak[4]
for particular words $X$ and $Y$, seeding the infinite repetitions
$\repleft X$ and $\repright Y$.  These words remain fixed for
different $\codedtaginput$ (possible since $R$ emulates a fixed
universal cyclic tag system).\pagebreak[4]

We shall construct words $\leftseed X$, $\leftstem X$, $\rightstem Y$
and $\rightseed Y$ (by the \defn{wrap construction} introduced below)
such that for all input words $I$, the Turing machine $M$ with initial
tape
$$\repleft {\leftseed X}
\;\;
\leftstem{ X}
\:\;I\;\,\rightstem Y\,\;\repright {\rightseed Y}$$
emulates $R$ with initial state
\begin{center}\vspace*{-2ex}
$\repleft X \;\;\codedtaginput \;\;\repright Y$
\vspace*{-.3ex}\end{center}
on the causal future of $I$ in the space-time diagram of $R$.  For example,
with\footnote{Here $X$ and $Y$ are not the words used by Wolfram and
Cook to prove the universality of $R$.  For illustrative purposes,
simpler words are used in this example.  Our emulation works for
\emph{any} $X$ and $Y$ that do not evolve to a sequence of all 0s (including the Wolfram-Cook words).}
\begin{displaymath}
X \;\;=\;\; 111011 \hspace*{9ex} I \;\;=\;\; 10011 \hspace*{9ex} Y
\;\;=\;\; 1101
\end{displaymath}
the evolution of $R$ begins 
\begin{center}\small$\cellsleft
\fc\tinyone \fc\tinyone \fc\tinyone \fc\tinyzero \fc\tinyone \fc\tinyone \fc\tinyone \fc\tinyone \fc\tinyone \fc\tinyzero \fc\tinyone \fc\tinyone \fc\tinyone \fc\tinyone \fc\tinyone \fc\tinyzero \fc\tinyone \fc\tinyone \fc\tinyone \fc\tinyone \fc\tinyone \fc\tinyzero \fc\tinyone \fc\tinyone \fc\one \fc\zero \fc\zero \fc\one \fc\one \fc\tinyone \fc\tinyone \fc\tinyzero \fc\tinyone \fc\tinyone \fc\tinyone \fc\tinyzero \fc\tinyone \fc\tinyone \fc\tinyone \fc\tinyzero \fc\tinyone \fc\tinyone \fc\tinyone \fc\tinyzero \fc\tinyone \fc\tinyone \fc\tinyone \fc\tinyzero \fc\tinyone \fc\tinyone \fc\tinyone \fc\tinyzero \fc\tinyone \\ 
 
\fc{}\fc\tinyzero \fc\tinyone \fc\tinyone \fc\tinyone \fc\tinyzero \fc\tinyzero \fc\tinyzero \fc\tinyone \fc\tinyone \fc\tinyone \fc\tinyzero \fc\tinyzero \fc\tinyzero \fc\tinyone \fc\tinyone \fc\tinyone \fc\tinyzero \fc\tinyzero \fc\tinyzero \fc\tinyone \fc\tinyone \fc\tinyone \fc\zero \fc\one \fc\zero \fc\one \fc\one \fc\zero \fc\zero \fc\tinyone \fc\tinyone \fc\tinyone \fc\tinyzero \fc\tinyone \fc\tinyone \fc\tinyone \fc\tinyzero \fc\tinyone \fc\tinyone \fc\tinyone \fc\tinyzero \fc\tinyone \fc\tinyone \fc\tinyone \fc\tinyzero \fc\tinyone \fc\tinyone \fc\tinyone \fc\tinyzero \fc\tinyone \fc\tinyone \\ 
 
\fc{}\fc{}\fc\tinyone \fc\tinyzero \fc\tinyone \fc\tinyzero \fc\tinyzero \fc\tinyone \fc\tinyone \fc\tinyzero \fc\tinyone \fc\tinyzero \fc\tinyzero \fc\tinyone \fc\tinyone \fc\tinyzero \fc\tinyone \fc\tinyzero \fc\tinyzero \fc\tinyone \fc\tinyone \fc\tinyzero \fc\one \fc\one \fc\one \fc\one \fc\one \fc\one \fc\zero \fc\one \fc\one \fc\tinyzero \fc\tinyone \fc\tinyone \fc\tinyone \fc\tinyzero \fc\tinyone \fc\tinyone \fc\tinyone \fc\tinyzero \fc\tinyone \fc\tinyone \fc\tinyone \fc\tinyzero \fc\tinyone \fc\tinyone \fc\tinyone \fc\tinyzero \fc\tinyone \fc\tinyone \fc\tinyone \\ 
 
\fc{}\fc{}\fc{}\fc\tinyone \fc\tinyone \fc\tinyzero \fc\tinyone \fc\tinyone \fc\tinyone \fc\tinyone \fc\tinyone \fc\tinyzero \fc\tinyone \fc\tinyone \fc\tinyone \fc\tinyone \fc\tinyone \fc\tinyzero \fc\tinyone \fc\tinyone \fc\tinyone \fc\one \fc\one \fc\zero \fc\zero \fc\zero \fc\zero \fc\one \fc\one \fc\one \fc\one \fc\one \fc\tinyone \fc\tinyzero \fc\tinyone \fc\tinyone \fc\tinyone \fc\tinyzero \fc\tinyone \fc\tinyone \fc\tinyone \fc\tinyzero \fc\tinyone \fc\tinyone \fc\tinyone \fc\tinyzero \fc\tinyone \fc\tinyone \fc\tinyone \fc\tinyzero \\ 
 
\fc{}\fc{}\fc{}\fc{}\fc\tinyone \fc\tinyone \fc\tinyone \fc\tinyzero \fc\tinyzero \fc\tinyzero \fc\tinyone \fc\tinyone \fc\tinyone \fc\tinyzero \fc\tinyzero \fc\tinyzero \fc\tinyone \fc\tinyone \fc\tinyone \fc\tinyzero \fc\zero \fc\zero \fc\one \fc\zero \fc\zero \fc\zero \fc\one \fc\one \fc\zero \fc\zero \fc\zero \fc\zero \fc\one \fc\tinyone \fc\tinyone \fc\tinyzero \fc\tinyone \fc\tinyone \fc\tinyone \fc\tinyzero \fc\tinyone \fc\tinyone \fc\tinyone \fc\tinyzero \fc\tinyone \fc\tinyone \fc\tinyone \fc\tinyzero \fc\tinyone \\ 
 
\fc{}\fc{}\fc{}\fc{}\fc{}\fc\tinyzero \fc\tinyone \fc\tinyzero \fc\tinyzero \fc\tinyone \fc\tinyone \fc\tinyzero \fc\tinyone \fc\tinyzero \fc\tinyzero \fc\tinyone \fc\tinyone \fc\tinyzero \fc\tinyone \fc\zero \fc\zero \fc\one \fc\one \fc\zero \fc\zero \fc\one \fc\one \fc\one \fc\zero \fc\zero \fc\zero \fc\one \fc\one \fc\zero \fc\tinyone \fc\tinyone \fc\tinyone \fc\tinyzero \fc\tinyone \fc\tinyone \fc\tinyone \fc\tinyzero \fc\tinyone \fc\tinyone \fc\tinyone \fc\tinyzero \fc\tinyone \fc\tinyone \\ 
 
\fc{}\fc{}\fc{}\fc{}\fc{}\fc{}\fc\tinyone \fc\tinyzero \fc\tinyone \fc\tinyone \fc\tinyone \fc\tinyone \fc\tinyone \fc\tinyzero \fc\tinyone \fc\tinyone \fc\tinyone \fc\tinyone \fc\one \fc\zero \fc\one \fc\one \fc\one \fc\zero \fc\one \fc\one \fc\zero \fc\one \fc\zero \fc\zero \fc\one \fc\one \fc\one \fc\one \fc\one \fc\tinyzero \fc\tinyone \fc\tinyone \fc\tinyone \fc\tinyzero \fc\tinyone \fc\tinyone \fc\tinyone \fc\tinyzero \fc\tinyone \fc\tinyone \fc\tinyone \\ 
 
\fc{}\fc{}\fc{}\fc{}\fc{}\fc{}\fc{}\fc\tinyone \fc\tinyone \fc\tinyzero \fc\tinyzero \fc\tinyzero \fc\tinyone \fc\tinyone \fc\tinyone \fc\tinyzero \fc\tinyzero \fc\zero \fc\one \fc\one \fc\one \fc\zero \fc\one \fc\one \fc\one \fc\one \fc\one \fc\one \fc\zero \fc\one \fc\one \fc\zero \fc\zero \fc\zero \fc\one \fc\one \fc\tinyone \fc\tinyzero \fc\tinyone \fc\tinyone \fc\tinyone \fc\tinyzero \fc\tinyone \fc\tinyone \fc\tinyone \fc\tinyzero \\ 
 
\fc{}\fc{}\fc{}\fc{}\fc{}\fc{}\fc{}\fc{}\fc\tinyone \fc\tinyzero \fc\tinyzero \fc\tinyone \fc\tinyone \fc\tinyzero \fc\tinyone \fc\tinyzero \fc\zero \fc\one \fc\one \fc\zero \fc\one \fc\one \fc\one \fc\zero \fc\zero \fc\zero \fc\zero \fc\one \fc\one \fc\one \fc\one \fc\zero \fc\zero \fc\one \fc\one \fc\zero \fc\one \fc\tinyone \fc\tinyone \fc\tinyzero \fc\tinyone \fc\tinyone \fc\tinyone \fc\tinyzero \fc\tinyone \\ 
 
\fc{}\fc{}\fc{}\fc{}\fc{}\fc{}\fc{}\fc{}\fc{}\fc\tinyzero \fc\tinyone \fc\tinyone \fc\tinyone \fc\tinyone \fc\tinyone \fc\zero \fc\one \fc\one \fc\one \fc\one \fc\one \fc\zero \fc\one \fc\zero \fc\zero \fc\zero \fc\one \fc\one \fc\zero \fc\zero \fc\one \fc\zero \fc\one \fc\one \fc\one \fc\one \fc\one \fc\zero \fc\tinyone \fc\tinyone \fc\tinyone \fc\tinyzero \fc\tinyone \fc\tinyone \\ 
 
\fc{}\fc{}\fc{}\fc{}\fc{}\fc{}\fc{}\fc{}\fc{}\fc{}\fc\tinyone \fc\tinyzero \fc\tinyzero \fc\tinyzero \fc\one \fc\one \fc\one \fc\zero \fc\zero \fc\zero \fc\one \fc\one \fc\one \fc\zero \fc\zero \fc\one \fc\one \fc\one \fc\zero \fc\one \fc\one \fc\one \fc\one \fc\zero \fc\zero \fc\zero \fc\one \fc\one \fc\one \fc\tinyzero \fc\tinyone \fc\tinyone \fc\tinyone \\ 
 
\fc{}\fc{}\fc{}\fc{}\fc{}\fc{}\fc{}\fc{}\fc{}\fc{}\fc{}\fc\tinyzero \fc\tinyzero \fc\one \fc\one \fc\zero \fc\one \fc\zero \fc\zero \fc\one \fc\one \fc\zero \fc\one \fc\zero \fc\one \fc\one \fc\zero \fc\one \fc\one \fc\one \fc\zero \fc\zero \fc\one \fc\zero \fc\zero \fc\one \fc\one \fc\zero \fc\one \fc\one \fc\tinyone \fc\tinyzero \\ 
 
\fc{}\fc{}\fc{}\fc{}\fc{}\fc{}\fc{}\fc{}\fc{}\fc{}\fc{}\fc{}\fc\one \fc\one \fc\one \fc\one \fc\one \fc\zero \fc\one \fc\one \fc\one \fc\one \fc\one \fc\one \fc\one \fc\one \fc\one \fc\one \fc\zero \fc\one \fc\zero \fc\one \fc\one \fc\zero \fc\one \fc\one \fc\one \fc\one \fc\one \fc\zero \fc\one \\ 
 
\cellsend$\end{center}
where larger symbols emphasise the causal future of $I$.
The wrap construction will turn out to yield
\begin{displaymath}
\leftseed  X \;\;=\;\; 11\varone 0 \varzero 0\varone 011 \hspace*{7ex}
\leftstem   X \text{ empty}
\hspace*{7ex}
\rightstem  Y \;\;=\;\; \varone\varone \hspace*{7ex}
\rightseed Y \;\;=\;\; \varzero 1 
\end{displaymath}
so we start the Turing machine $M$ with initial tape
$\;\repleft {\leftseed X}
\;
\leftstem{X}
\;I\,\rightstem Y\,\repright {\rightseed Y}\,$, a segment of which is
\begin{center}\vspace*{-.5ex}\small$\cellsleft
\tc\tinyone \tc\tinyone \tc\tinyvarone \tc\tinyzero \tc\tinyvarzero \tc\tinyzero \tc\tinyvarone \tc\tinyzero \tc\tinyone \tc\tinyone \tc\tinyone \tc\tinyone \tc\tinyvarone \tc\tinyzero \tc\tinyvarzero \tc\tinyzero \tc\tinyvarone \tc\tinyzero \tc\tinyone \tc\tinyone \tc\tinyone \tc\tinyone \tc\tinyvarone \tc\tinyzero \tc\tinyvarzero \tc\tinyzero \tc\tinyvarone \tc\tinyzero \tc\tinyone \tc\tinyone \tc\tinyone \tc\tinyone \tc\tinyvarone \tc\tinyzero \tc\tinyvarzero \tc\tinyzero \tc\tinyvarone \tc\tinyzero \tc\tinyone \tc\tinyone \tc\one \tc\zero \tc\zero \tc\one \tc\one \nocarry{\tc\tinyvarone} \tc\tinyvarone \tc\tinyvarzero \tc\tinyone \tc\tinyvarzero \tc\tinyone \tc\tinyvarzero \tc\tinyone \tc\tinyvarzero \tc\tinyone \tc\tinyvarzero \tc\tinyone \tc\tinyvarzero \tc\tinyone \tc\tinyvarzero \tc\tinyone \tc\tinyvarzero \tc\tinyone \tc\tinyvarzero \tc\tinyone \tc\tinyvarzero \tc\tinyone \tc\tinyvarzero \tc\tinyone \tc\tinyvarzero \tc\tinyone \tc\tinyvarzero \tc\tinyone \tc\tinyvarzero \tc\tinyone \tc\tinyvarzero \tc\tinyone \tc\tinyvarzero \tc\tinyone \\
\cellsend$\end{center}
where $I$ is emphasised with larger symbols.  The Turing machine $M$
evolves as follows:
\begin{center}\vspace*{-.5ex}\small$\label{emulation}\renewcommand{\arraystretch}{.95}\cellsleft
\tc\tinyone \tc\tinyone \tc\tinyvarone \tc\tinyzero \tc\tinyvarzero \tc\tinyzero \tc\tinyvarone \tc\tinyzero \tc\tinyone \tc\tinyone \tc\tinyone \tc\tinyone \tc\tinyvarone \tc\tinyzero \tc\tinyvarzero \tc\tinyzero \tc\tinyvarone \tc\tinyzero \tc\tinyone \tc\tinyone \tc\tinyone \tc\tinyone \tc\tinyvarone \tc\tinyzero \tc\tinyvarzero \tc\tinyzero \tc\tinyvarone \tc\tinyzero \tc\tinyone \tc\tinyone \tc\tinyone \tc\tinyone \tc\tinyvarone \tc\tinyzero \tc\tinyvarzero \tc\tinyzero \tc\tinyvarone \tc\tinyzero \tc\tinyone \tc\tinyone \tc\one \tc\zero \tc\zero \tc\one \tc\one \tc\tinyone \tc\tinyone \tc\tinyzero \nocarry{\tc\tinyone} \tc\tinyvarzero \tc\tinyone \tc\tinyvarzero \tc\tinyone \tc\tinyvarzero \tc\tinyone \tc\tinyvarzero \tc\tinyone \tc\tinyvarzero \tc\tinyone \tc\tinyvarzero \tc\tinyone \tc\tinyvarzero \tc\tinyone \tc\tinyvarzero \tc\tinyone \tc\tinyvarzero \tc\tinyone \tc\tinyvarzero \tc\tinyone \tc\tinyvarzero \tc\tinyone \tc\tinyvarzero \tc\tinyone \tc\tinyvarzero \tc\tinyone \tc\tinyvarzero \tc\tinyone \tc\tinyvarzero \tc\tinyone \\
\tc\tinyone \tc\tinyone \tc\tinyvarone \tc\tinyzero \tc\tinyvarzero \tc\tinyzero \tc\tinyvarone \tc\tinyzero \tc\tinyone \tc\tinyone \tc\tinyone \tc\tinyone \tc\tinyvarone \tc\tinyzero \tc\tinyvarzero \tc\tinyzero \tc\tinyvarone \tc\tinyzero \tc\tinyone \tc\tinyone \tc\tinyone \tc\tinyone \tc\tinyvarone \tc\tinyzero \tc\tinyvarzero \tc\tinyzero \tc\tinyvarone \tc\tinyzero \tc\tinyone \tc\tinyone \tc\tinyone \tc\tinyone \tc\tinyvarone \tc\tinyzero \tc\tinyvarzero \tc\tinyzero \tc\tinyone \tc\tinyone \tc\tinyone \tc\zero \tc\one \tc\zero \tc\one \tc\one \tc\zero \tc\zero \tc\tinyone \tc\tinyone \tc\tinyone \tc\tinyzero \nocarry{\tc\tinyone} \tc\tinyvarzero \tc\tinyone \tc\tinyvarzero \tc\tinyone \tc\tinyvarzero \tc\tinyone \tc\tinyvarzero \tc\tinyone \tc\tinyvarzero \tc\tinyone \tc\tinyvarzero \tc\tinyone \tc\tinyvarzero \tc\tinyone \tc\tinyvarzero \tc\tinyone \tc\tinyvarzero \tc\tinyone \tc\tinyvarzero \tc\tinyone \tc\tinyvarzero \tc\tinyone \tc\tinyvarzero \tc\tinyone \tc\tinyvarzero \tc\tinyone \tc\tinyvarzero \tc\tinyone \\
\tc\tinyone \tc\tinyone \tc\tinyvarone \tc\tinyzero \tc\tinyvarzero \tc\tinyzero \tc\tinyvarone \tc\tinyzero \tc\tinyone \tc\tinyone \tc\tinyone \tc\tinyone \tc\tinyvarone \tc\tinyzero \tc\tinyvarzero \tc\tinyzero \tc\tinyvarone \tc\tinyzero \tc\tinyone \tc\tinyone \tc\tinyone \tc\tinyone \tc\tinyvarone \tc\tinyzero \tc\tinyvarzero \tc\tinyzero \tc\tinyvarone \tc\tinyzero \tc\tinyone \tc\tinyone \tc\tinyone \tc\tinyone \tc\tinyvarone \tc\tinyzero \tc\tinyzero \tc\tinyone \tc\tinyone \tc\tinyzero \tc\one \tc\one \tc\one \tc\one \tc\one \tc\one \tc\zero \tc\one \tc\one \tc\tinyzero \tc\tinyone \tc\tinyone \tc\tinyone \tc\tinyzero \nocarry{\tc\tinyone} \tc\tinyvarzero \tc\tinyone \tc\tinyvarzero \tc\tinyone \tc\tinyvarzero \tc\tinyone \tc\tinyvarzero \tc\tinyone \tc\tinyvarzero \tc\tinyone \tc\tinyvarzero \tc\tinyone \tc\tinyvarzero \tc\tinyone \tc\tinyvarzero \tc\tinyone \tc\tinyvarzero \tc\tinyone \tc\tinyvarzero \tc\tinyone \tc\tinyvarzero \tc\tinyone \tc\tinyvarzero \tc\tinyone \tc\tinyvarzero \tc\tinyone \\
\tc\tinyone \tc\tinyone \tc\tinyvarone \tc\tinyzero \tc\tinyvarzero \tc\tinyzero \tc\tinyvarone \tc\tinyzero \tc\tinyone \tc\tinyone \tc\tinyone \tc\tinyone \tc\tinyvarone \tc\tinyzero \tc\tinyvarzero \tc\tinyzero \tc\tinyvarone \tc\tinyzero \tc\tinyone \tc\tinyone \tc\tinyone \tc\tinyone \tc\tinyvarone \tc\tinyzero \tc\tinyvarzero \tc\tinyzero \tc\tinyvarone \tc\tinyzero \tc\tinyone \tc\tinyone \tc\tinyone \tc\tinyone \tc\tinyone \tc\tinyzero \tc\tinyone \tc\tinyone \tc\tinyone \tc\one \tc\one \tc\zero \tc\zero \tc\zero \tc\zero \tc\one \tc\one \tc\one \tc\one \tc\one \tc\tinyone \tc\tinyzero \tc\tinyone \tc\tinyone \tc\tinyone \tc\tinyzero \nocarry{\tc\tinyone} \tc\tinyvarzero \tc\tinyone \tc\tinyvarzero \tc\tinyone \tc\tinyvarzero \tc\tinyone \tc\tinyvarzero \tc\tinyone \tc\tinyvarzero \tc\tinyone \tc\tinyvarzero \tc\tinyone \tc\tinyvarzero \tc\tinyone \tc\tinyvarzero \tc\tinyone \tc\tinyvarzero \tc\tinyone \tc\tinyvarzero \tc\tinyone \tc\tinyvarzero \tc\tinyone \tc\tinyvarzero \tc\tinyone \\
\tc\tinyone \tc\tinyone \tc\tinyvarone \tc\tinyzero \tc\tinyvarzero \tc\tinyzero \tc\tinyvarone \tc\tinyzero \tc\tinyone \tc\tinyone \tc\tinyone \tc\tinyone \tc\tinyvarone \tc\tinyzero \tc\tinyvarzero \tc\tinyzero \tc\tinyvarone \tc\tinyzero \tc\tinyone \tc\tinyone \tc\tinyone \tc\tinyone \tc\tinyvarone \tc\tinyzero \tc\tinyvarzero \tc\tinyzero \tc\tinyone \tc\tinyone \tc\tinyone \tc\tinyzero \tc\tinyzero \tc\tinyzero \tc\tinyone \tc\tinyone \tc\tinyone \tc\tinyzero \tc\zero \tc\zero \tc\one \tc\zero \tc\zero \tc\zero \tc\one \tc\one \tc\zero \tc\zero \tc\zero \tc\zero \tc\one \tc\tinyone \tc\tinyone \tc\tinyzero \tc\tinyone \tc\tinyone \tc\tinyone \tc\tinyzero \nocarry{\tc\tinyone} \tc\tinyvarzero \tc\tinyone \tc\tinyvarzero \tc\tinyone \tc\tinyvarzero \tc\tinyone \tc\tinyvarzero \tc\tinyone \tc\tinyvarzero \tc\tinyone \tc\tinyvarzero \tc\tinyone \tc\tinyvarzero \tc\tinyone \tc\tinyvarzero \tc\tinyone \tc\tinyvarzero \tc\tinyone \tc\tinyvarzero \tc\tinyone \tc\tinyvarzero \tc\tinyone \\
\tc\tinyone \tc\tinyone \tc\tinyvarone \tc\tinyzero \tc\tinyvarzero \tc\tinyzero \tc\tinyvarone \tc\tinyzero \tc\tinyone \tc\tinyone \tc\tinyone \tc\tinyone \tc\tinyvarone \tc\tinyzero \tc\tinyvarzero \tc\tinyzero \tc\tinyvarone \tc\tinyzero \tc\tinyone \tc\tinyone \tc\tinyone \tc\tinyone \tc\tinyvarone \tc\tinyzero \tc\tinyzero \tc\tinyone \tc\tinyone \tc\tinyzero \tc\tinyone \tc\tinyzero \tc\tinyzero \tc\tinyone \tc\tinyone \tc\tinyzero \tc\tinyone \tc\zero \tc\zero \tc\one \tc\one \tc\zero \tc\zero \tc\one \tc\one \tc\one \tc\zero \tc\zero \tc\zero \tc\one \tc\one \tc\zero \tc\tinyone \tc\tinyone \tc\tinyone \tc\tinyzero \tc\tinyone \tc\tinyone \tc\tinyone \tc\tinyzero \nocarry{\tc\tinyone} \tc\tinyvarzero \tc\tinyone \tc\tinyvarzero \tc\tinyone \tc\tinyvarzero \tc\tinyone \tc\tinyvarzero \tc\tinyone \tc\tinyvarzero \tc\tinyone \tc\tinyvarzero \tc\tinyone \tc\tinyvarzero \tc\tinyone \tc\tinyvarzero \tc\tinyone \tc\tinyvarzero \tc\tinyone \tc\tinyvarzero \tc\tinyone \\
\tc\tinyone \tc\tinyone \tc\tinyvarone \tc\tinyzero \tc\tinyvarzero \tc\tinyzero \tc\tinyvarone \tc\tinyzero \tc\tinyone \tc\tinyone \tc\tinyone \tc\tinyone \tc\tinyvarone \tc\tinyzero \tc\tinyvarzero \tc\tinyzero \tc\tinyvarone \tc\tinyzero \tc\tinyone \tc\tinyone \tc\tinyone \tc\tinyone \tc\tinyone \tc\tinyzero \tc\tinyone \tc\tinyone \tc\tinyone \tc\tinyone \tc\tinyone \tc\tinyzero \tc\tinyone \tc\tinyone \tc\tinyone \tc\tinyone \tc\one \tc\zero \tc\one \tc\one \tc\one \tc\zero \tc\one \tc\one \tc\zero \tc\one \tc\zero \tc\zero \tc\one \tc\one \tc\one \tc\one \tc\one \tc\tinyzero \tc\tinyone \tc\tinyone \tc\tinyone \tc\tinyzero \tc\tinyone \tc\tinyone \tc\tinyone \tc\tinyzero \nocarry{\tc\tinyone} \tc\tinyvarzero \tc\tinyone \tc\tinyvarzero \tc\tinyone \tc\tinyvarzero \tc\tinyone \tc\tinyvarzero \tc\tinyone \tc\tinyvarzero \tc\tinyone \tc\tinyvarzero \tc\tinyone \tc\tinyvarzero \tc\tinyone \tc\tinyvarzero \tc\tinyone \tc\tinyvarzero \tc\tinyone \\
\tc\tinyone \tc\tinyone \tc\tinyvarone \tc\tinyzero \tc\tinyvarzero \tc\tinyzero \tc\tinyvarone \tc\tinyzero \tc\tinyone \tc\tinyone \tc\tinyone \tc\tinyone \tc\tinyvarone \tc\tinyzero \tc\tinyvarzero \tc\tinyzero \tc\tinyone \tc\tinyone \tc\tinyone \tc\tinyzero \tc\tinyzero \tc\tinyzero \tc\tinyone \tc\tinyone \tc\tinyone \tc\tinyzero \tc\tinyzero \tc\tinyzero \tc\tinyone \tc\tinyone \tc\tinyone \tc\tinyzero \tc\tinyzero \tc\zero \tc\one \tc\one \tc\one \tc\zero \tc\one \tc\one \tc\one \tc\one \tc\one \tc\one \tc\zero \tc\one \tc\one \tc\zero \tc\zero \tc\zero \tc\one \tc\one \tc\tinyone \tc\tinyzero \tc\tinyone \tc\tinyone \tc\tinyone \tc\tinyzero \tc\tinyone \tc\tinyone \tc\tinyone \tc\tinyzero \nocarry{\tc\tinyone} \tc\tinyvarzero \tc\tinyone \tc\tinyvarzero \tc\tinyone \tc\tinyvarzero \tc\tinyone \tc\tinyvarzero \tc\tinyone \tc\tinyvarzero \tc\tinyone \tc\tinyvarzero \tc\tinyone \tc\tinyvarzero \tc\tinyone \tc\tinyvarzero \tc\tinyone \\
\tc\tinyone \tc\tinyone \tc\tinyvarone \tc\tinyzero \tc\tinyvarzero \tc\tinyzero \tc\tinyvarone \tc\tinyzero \tc\tinyone \tc\tinyone \tc\tinyone \tc\tinyone \tc\tinyvarone \tc\tinyzero \tc\tinyzero \tc\tinyone \tc\tinyone \tc\tinyzero \tc\tinyone \tc\tinyzero \tc\tinyzero \tc\tinyone \tc\tinyone \tc\tinyzero \tc\tinyone \tc\tinyzero \tc\tinyzero \tc\tinyone \tc\tinyone \tc\tinyzero \tc\tinyone \tc\tinyzero \tc\zero \tc\one \tc\one \tc\zero \tc\one \tc\one \tc\one \tc\zero \tc\zero \tc\zero \tc\zero \tc\one \tc\one \tc\one \tc\one \tc\zero \tc\zero \tc\one \tc\one \tc\zero \tc\one \tc\tinyone \tc\tinyone \tc\tinyzero \tc\tinyone \tc\tinyone \tc\tinyone \tc\tinyzero \tc\tinyone \tc\tinyone \tc\tinyone \tc\tinyzero \nocarry{\tc\tinyone} \tc\tinyvarzero \tc\tinyone \tc\tinyvarzero \tc\tinyone \tc\tinyvarzero \tc\tinyone \tc\tinyvarzero \tc\tinyone \tc\tinyvarzero \tc\tinyone \tc\tinyvarzero \tc\tinyone \tc\tinyvarzero \tc\tinyone \\
\tc\tinyone \tc\tinyone \tc\tinyvarone \tc\tinyzero \tc\tinyvarzero \tc\tinyzero \tc\tinyvarone \tc\tinyzero \tc\tinyone \tc\tinyone \tc\tinyone \tc\tinyone \tc\tinyone \tc\tinyzero \tc\tinyone \tc\tinyone \tc\tinyone \tc\tinyone \tc\tinyone \tc\tinyzero \tc\tinyone \tc\tinyone \tc\tinyone \tc\tinyone \tc\tinyone \tc\tinyzero \tc\tinyone \tc\tinyone \tc\tinyone \tc\tinyone \tc\tinyone \tc\zero \tc\one \tc\one \tc\one \tc\one \tc\one \tc\zero \tc\one \tc\zero \tc\zero \tc\zero \tc\one \tc\one \tc\zero \tc\zero \tc\one \tc\zero \tc\one \tc\one \tc\one \tc\one \tc\one \tc\zero \tc\tinyone \tc\tinyone \tc\tinyone \tc\tinyzero \tc\tinyone \tc\tinyone \tc\tinyone \tc\tinyzero \tc\tinyone \tc\tinyone \tc\tinyone \tc\tinyzero \nocarry{\tc\tinyone} \tc\tinyvarzero \tc\tinyone \tc\tinyvarzero \tc\tinyone \tc\tinyvarzero \tc\tinyone \tc\tinyvarzero \tc\tinyone \tc\tinyvarzero \tc\tinyone \tc\tinyvarzero \tc\tinyone \\
\tc\tinyone \tc\tinyone \tc\tinyvarone \tc\tinyzero \tc\tinyvarzero \tc\tinyzero \tc\tinyone \tc\tinyone \tc\tinyone \tc\tinyzero \tc\tinyzero \tc\tinyzero \tc\tinyone \tc\tinyone \tc\tinyone \tc\tinyzero \tc\tinyzero \tc\tinyzero \tc\tinyone \tc\tinyone \tc\tinyone \tc\tinyzero \tc\tinyzero \tc\tinyzero \tc\tinyone \tc\tinyone \tc\tinyone \tc\tinyzero \tc\tinyzero \tc\tinyzero \tc\one \tc\one \tc\one \tc\zero \tc\zero \tc\zero \tc\one \tc\one \tc\one \tc\zero \tc\zero \tc\one \tc\one \tc\one \tc\zero \tc\one \tc\one \tc\one \tc\one \tc\zero \tc\zero \tc\zero \tc\one \tc\one \tc\one \tc\tinyzero \tc\tinyone \tc\tinyone \tc\tinyone \tc\tinyzero \tc\tinyone \tc\tinyone \tc\tinyone \tc\tinyzero \tc\tinyone \tc\tinyone \tc\tinyone \tc\tinyzero \nocarry{\tc\tinyone} \tc\tinyvarzero \tc\tinyone \tc\tinyvarzero \tc\tinyone \tc\tinyvarzero \tc\tinyone \tc\tinyvarzero \tc\tinyone \tc\tinyvarzero \tc\tinyone \\
\tc\tinyone \tc\tinyone \tc\tinyvarone \tc\tinyzero \tc\tinyzero \tc\tinyone \tc\tinyone \tc\tinyzero \tc\tinyone \tc\tinyzero \tc\tinyzero \tc\tinyone \tc\tinyone \tc\tinyzero \tc\tinyone \tc\tinyzero \tc\tinyzero \tc\tinyone \tc\tinyone \tc\tinyzero \tc\tinyone \tc\tinyzero \tc\tinyzero \tc\tinyone \tc\tinyone \tc\tinyzero \tc\tinyone \tc\tinyzero \tc\tinyzero \tc\one \tc\one \tc\zero \tc\one \tc\zero \tc\zero \tc\one \tc\one \tc\zero \tc\one \tc\zero \tc\one \tc\one \tc\zero \tc\one \tc\one \tc\one \tc\zero \tc\zero \tc\one \tc\zero \tc\zero \tc\one \tc\one \tc\zero \tc\one \tc\one \tc\tinyone \tc\tinyzero \tc\tinyone \tc\tinyone \tc\tinyone \tc\tinyzero \tc\tinyone \tc\tinyone \tc\tinyone \tc\tinyzero \tc\tinyone \tc\tinyone \tc\tinyone \tc\tinyzero \nocarry{\tc\tinyone} \tc\tinyvarzero \tc\tinyone \tc\tinyvarzero \tc\tinyone \tc\tinyvarzero \tc\tinyone \tc\tinyvarzero \tc\tinyone \\
\tc\tinyone \tc\tinyone \tc\tinyone \tc\tinyzero \tc\tinyone \tc\tinyone \tc\tinyone \tc\tinyone \tc\tinyone \tc\tinyzero \tc\tinyone \tc\tinyone \tc\tinyone \tc\tinyone \tc\tinyone \tc\tinyzero \tc\tinyone \tc\tinyone \tc\tinyone \tc\tinyone \tc\tinyone \tc\tinyzero \tc\tinyone \tc\tinyone \tc\tinyone \tc\tinyone \tc\tinyone \tc\tinyzero \tc\one \tc\one \tc\one \tc\one \tc\one \tc\zero \tc\one \tc\one \tc\one \tc\one \tc\one \tc\one \tc\one \tc\one \tc\one \tc\one \tc\zero \tc\one \tc\zero \tc\one \tc\one \tc\zero \tc\one \tc\one \tc\one \tc\one \tc\one \tc\zero \tc\one \tc\tinyone \tc\tinyone \tc\tinyzero \tc\tinyone \tc\tinyone \tc\tinyone \tc\tinyzero \tc\tinyone \tc\tinyone \tc\tinyone \tc\tinyzero \tc\tinyone \tc\tinyone \tc\tinyone \tc\tinyzero \nocarry{\tc\tinyone} \tc\tinyvarzero \tc\tinyone \tc\tinyvarzero \tc\tinyone \tc\tinyvarzero \tc\tinyone \\
\cellsend$\end{center}
The head makes progressively larger left and
right sweeps, analogous to
Wolfram and Cook's limited emulation on 
$\repleft\varzero\; 0\; \tapedata\; \repright 0\,$.
The tape is
shown at the end of every rightward sweep.
Larger symbols emphasise the causal future of $I$. 
Note that, on the causal future of $I$, the Turing machine has indeed
emulated the cellular automaton.
Below we have interleaved the ends of the leftward sweeps:
\begin{center}\label{leftright}\small$\renewcommand{\arraystretch}{.9}\cellsleft\label{interleaved}
\tc\tinyone \tc\tinyone \tc\tinyvarone \tc\tinyzero \tc\tinyvarzero \tc\tinyzero \tc\tinyvarone \tc\tinyzero \tc\tinyone \tc\tinyone \tc\tinyone \tc\tinyone \tc\tinyvarone \tc\tinyzero \tc\tinyvarzero \tc\tinyzero \tc\tinyvarone \tc\tinyzero \tc\tinyone \tc\tinyone \tc\tinyone \tc\tinyone \tc\tinyvarone \tc\tinyzero \tc\tinyvarzero \tc\tinyzero \tc\tinyvarone \tc\tinyzero \tc\tinyone \tc\tinyone \tc\tinyone \tc\tinyone \tc\tinyvarone \tc\tinyzero \tc\tinyvarzero \tc\tinyzero \tc\tinyvarone \tc\tinyzero \tc\tinyone \tc\tinyone \tc\one \tc\zero \tc\zero \tc\one \tc\one \tc\tinyone \tc\tinyone \tc\tinyzero \nocarry{\tc\tinyone} \tc\tinyvarzero \tc\tinyone \tc\tinyvarzero \tc\tinyone \tc\tinyvarzero \tc\tinyone \tc\tinyvarzero \tc\tinyone \tc\tinyvarzero \tc\tinyone \tc\tinyvarzero \tc\tinyone \tc\tinyvarzero \tc\tinyone \tc\tinyvarzero \tc\tinyone \tc\tinyvarzero \tc\tinyone \tc\tinyvarzero \tc\tinyone \tc\tinyvarzero \tc\tinyone \tc\tinyvarzero \tc\tinyone \tc\tinyvarzero \tc\tinyone \tc\tinyvarzero \tc\tinyone \tc\tinyvarzero \tc\tinyone \\
\tc\tinyone \tc\tinyone \tc\tinyvarone \tc\tinyzero \tc\tinyvarzero \tc\tinyzero \tc\tinyvarone \tc\tinyzero \tc\tinyone \tc\tinyone \tc\tinyone \tc\tinyone \tc\tinyvarone \tc\tinyzero \tc\tinyvarzero \tc\tinyzero \tc\tinyvarone \tc\tinyzero \tc\tinyone \tc\tinyone \tc\tinyone \tc\tinyone \tc\tinyvarone \tc\tinyzero \tc\tinyvarzero \tc\tinyzero \tc\tinyvarone \tc\tinyzero \tc\tinyone \tc\tinyone \tc\tinyone \tc\tinyone \tc\tinyvarone \tc\tinyzero \tc\tinyvarzero \tc\tinyzero \nocarry{\tc\tinyvarone} \tc\tinyvarone \tc\tinydunno \tc\dunno \tc\varone \tc\varzero \tc\varone \tc\dunno \tc\dunno \tc\dunno \tc\tinyvarone \tc\tinyvarone \tc\tinyvarone \tc\tinyvarzero \tc\tinyone \tc\tinyvarzero \tc\tinyone \tc\tinyvarzero \tc\tinyone \tc\tinyvarzero \tc\tinyone \tc\tinyvarzero \tc\tinyone \tc\tinyvarzero \tc\tinyone \tc\tinyvarzero \tc\tinyone \tc\tinyvarzero \tc\tinyone \tc\tinyvarzero \tc\tinyone \tc\tinyvarzero \tc\tinyone \tc\tinyvarzero \tc\tinyone \tc\tinyvarzero \tc\tinyone \tc\tinyvarzero \tc\tinyone \tc\tinyvarzero \tc\tinyone \tc\tinyvarzero \tc\tinyone \\
\tc\tinyone \tc\tinyone \tc\tinyvarone \tc\tinyzero \tc\tinyvarzero \tc\tinyzero \tc\tinyvarone \tc\tinyzero \tc\tinyone \tc\tinyone \tc\tinyone \tc\tinyone \tc\tinyvarone \tc\tinyzero \tc\tinyvarzero \tc\tinyzero \tc\tinyvarone \tc\tinyzero \tc\tinyone \tc\tinyone \tc\tinyone \tc\tinyone \tc\tinyvarone \tc\tinyzero \tc\tinyvarzero \tc\tinyzero \tc\tinyvarone \tc\tinyzero \tc\tinyone \tc\tinyone \tc\tinyone \tc\tinyone \tc\tinyvarone \tc\tinyzero \tc\tinyvarzero \tc\tinyzero \tc\tinyone \tc\tinyone \tc\tinyone \tc\zero \tc\one \tc\zero \tc\one \tc\one \tc\zero \tc\zero \tc\tinyone \tc\tinyone \tc\tinyone \tc\tinyzero \nocarry{\tc\tinyone} \tc\tinyvarzero \tc\tinyone \tc\tinyvarzero \tc\tinyone \tc\tinyvarzero \tc\tinyone \tc\tinyvarzero \tc\tinyone \tc\tinyvarzero \tc\tinyone \tc\tinyvarzero \tc\tinyone \tc\tinyvarzero \tc\tinyone \tc\tinyvarzero \tc\tinyone \tc\tinyvarzero \tc\tinyone \tc\tinyvarzero \tc\tinyone \tc\tinyvarzero \tc\tinyone \tc\tinyvarzero \tc\tinyone \tc\tinyvarzero \tc\tinyone \tc\tinyvarzero \tc\tinyone \\
\tc\tinyone \tc\tinyone \tc\tinyvarone \tc\tinyzero \tc\tinyvarzero \tc\tinyzero \tc\tinyvarone \tc\tinyzero \tc\tinyone \tc\tinyone \tc\tinyone \tc\tinyone \tc\tinyvarone \tc\tinyzero \tc\tinyvarzero \tc\tinyzero \tc\tinyvarone \tc\tinyzero \tc\tinyone \tc\tinyone \tc\tinyone \tc\tinyone \tc\tinyvarone \tc\tinyzero \tc\tinyvarzero \tc\tinyzero \tc\tinyvarone \tc\tinyzero \tc\tinyone \tc\tinyone \tc\tinyone \tc\tinyone \tc\tinyvarone \tc\tinyzero \nocarry{\tc\tinyvarzero} \tc\tinyvarone \tc\tinydunno \tc\tinydunno \tc\varone \tc\varone \tc\varone \tc\varone \tc\dunno \tc\varone \tc\varzero \tc\varone \tc\dunno \tc\tinydunno \tc\tinyvarone \tc\tinyvarone \tc\tinyvarone \tc\tinyvarzero \tc\tinyone \tc\tinyvarzero \tc\tinyone \tc\tinyvarzero \tc\tinyone \tc\tinyvarzero \tc\tinyone \tc\tinyvarzero \tc\tinyone \tc\tinyvarzero \tc\tinyone \tc\tinyvarzero \tc\tinyone \tc\tinyvarzero \tc\tinyone \tc\tinyvarzero \tc\tinyone \tc\tinyvarzero \tc\tinyone \tc\tinyvarzero \tc\tinyone \tc\tinyvarzero \tc\tinyone \tc\tinyvarzero \tc\tinyone \tc\tinyvarzero \tc\tinyone \\
\tc\tinyone \tc\tinyone \tc\tinyvarone \tc\tinyzero \tc\tinyvarzero \tc\tinyzero \tc\tinyvarone \tc\tinyzero \tc\tinyone \tc\tinyone \tc\tinyone \tc\tinyone \tc\tinyvarone \tc\tinyzero \tc\tinyvarzero \tc\tinyzero \tc\tinyvarone \tc\tinyzero \tc\tinyone \tc\tinyone \tc\tinyone \tc\tinyone \tc\tinyvarone \tc\tinyzero \tc\tinyvarzero \tc\tinyzero \tc\tinyvarone \tc\tinyzero \tc\tinyone \tc\tinyone \tc\tinyone \tc\tinyone \tc\tinyvarone \tc\tinyzero \tc\tinyzero \tc\tinyone \tc\tinyone \tc\tinyzero \tc\one \tc\one \tc\one \tc\one \tc\one \tc\one \tc\zero \tc\one \tc\one \tc\tinyzero \tc\tinyone \tc\tinyone \tc\tinyone \tc\tinyzero \nocarry{\tc\tinyone} \tc\tinyvarzero \tc\tinyone \tc\tinyvarzero \tc\tinyone \tc\tinyvarzero \tc\tinyone \tc\tinyvarzero \tc\tinyone \tc\tinyvarzero \tc\tinyone \tc\tinyvarzero \tc\tinyone \tc\tinyvarzero \tc\tinyone \tc\tinyvarzero \tc\tinyone \tc\tinyvarzero \tc\tinyone \tc\tinyvarzero \tc\tinyone \tc\tinyvarzero \tc\tinyone \tc\tinyvarzero \tc\tinyone \tc\tinyvarzero \tc\tinyone \\
\tc\tinyone \tc\tinyone \tc\tinyvarone \tc\tinyzero \tc\tinyvarzero \tc\tinyzero \tc\tinyvarone \tc\tinyzero \tc\tinyone \tc\tinyone \tc\tinyone \tc\tinyone \tc\tinyvarone \tc\tinyzero \tc\tinyvarzero \tc\tinyzero \tc\tinyvarone \tc\tinyzero \tc\tinyone \tc\tinyone \tc\tinyone \tc\tinyone \tc\tinyvarone \tc\tinyzero \tc\tinyvarzero \tc\tinyzero \tc\tinyvarone \tc\tinyzero \tc\tinyone \tc\tinyone \tc\tinyone \tc\tinyone \nocarry{\tc\tinyvarone} \tc\tinyvarzero \tc\tinyvarone \tc\tinydunno \tc\tinyvarone \tc\varone \tc\dunno \tc\dunno \tc\dunno \tc\dunno \tc\dunno \tc\varone \tc\varone \tc\dunno \tc\varone \tc\varone \tc\tinydunno \tc\tinydunno \tc\tinyvarone \tc\tinyvarone \tc\tinyvarone \tc\tinyvarzero \tc\tinyone \tc\tinyvarzero \tc\tinyone \tc\tinyvarzero \tc\tinyone \tc\tinyvarzero \tc\tinyone \tc\tinyvarzero \tc\tinyone \tc\tinyvarzero \tc\tinyone \tc\tinyvarzero \tc\tinyone \tc\tinyvarzero \tc\tinyone \tc\tinyvarzero \tc\tinyone \tc\tinyvarzero \tc\tinyone \tc\tinyvarzero \tc\tinyone \tc\tinyvarzero \tc\tinyone \tc\tinyvarzero \tc\tinyone \\
\tc\tinyone \tc\tinyone \tc\tinyvarone \tc\tinyzero \tc\tinyvarzero \tc\tinyzero \tc\tinyvarone \tc\tinyzero \tc\tinyone \tc\tinyone \tc\tinyone \tc\tinyone \tc\tinyvarone \tc\tinyzero \tc\tinyvarzero \tc\tinyzero \tc\tinyvarone \tc\tinyzero \tc\tinyone \tc\tinyone \tc\tinyone \tc\tinyone \tc\tinyvarone \tc\tinyzero \tc\tinyvarzero \tc\tinyzero \tc\tinyvarone \tc\tinyzero \tc\tinyone \tc\tinyone \tc\tinyone \tc\tinyone \tc\tinyone \tc\tinyzero \tc\tinyone \tc\tinyone \tc\tinyone \tc\one \tc\one \tc\zero \tc\zero \tc\zero \tc\zero \tc\one \tc\one \tc\one \tc\one \tc\one \tc\tinyone \tc\tinyzero \tc\tinyone \tc\tinyone \tc\tinyone \tc\tinyzero \nocarry{\tc\tinyone} \tc\tinyvarzero \tc\tinyone \tc\tinyvarzero \tc\tinyone \tc\tinyvarzero \tc\tinyone \tc\tinyvarzero \tc\tinyone \tc\tinyvarzero \tc\tinyone \tc\tinyvarzero \tc\tinyone \tc\tinyvarzero \tc\tinyone \tc\tinyvarzero \tc\tinyone \tc\tinyvarzero \tc\tinyone \tc\tinyvarzero \tc\tinyone \tc\tinyvarzero \tc\tinyone \tc\tinyvarzero \tc\tinyone \\
\tc\tinyone \tc\tinyone \tc\tinyvarone \tc\tinyzero \tc\tinyvarzero \tc\tinyzero \tc\tinyvarone \tc\tinyzero \tc\tinyone \tc\tinyone \tc\tinyone \tc\tinyone \tc\tinyvarone \tc\tinyzero \tc\tinyvarzero \tc\tinyzero \tc\tinyvarone \tc\tinyzero \tc\tinyone \tc\tinyone \tc\tinyone \tc\tinyone \tc\tinyvarone \tc\tinyzero \tc\tinyvarzero \tc\tinyzero \nocarry{\tc\tinyvarone} \tc\tinyvarone \tc\tinydunno \tc\tinydunno \tc\tinydunno \tc\tinydunno \tc\tinyvarone \tc\tinyvarone \tc\tinydunno \tc\tinydunno \tc\dunno \tc\dunno \tc\varone \tc\varzero \tc\varzero \tc\varzero \tc\varone \tc\dunno \tc\dunno \tc\dunno \tc\dunno \tc\dunno \tc\varone \tc\tinyvarone \tc\tinydunno \tc\tinydunno \tc\tinyvarone \tc\tinyvarone \tc\tinyvarone \tc\tinyvarzero \tc\tinyone \tc\tinyvarzero \tc\tinyone \tc\tinyvarzero \tc\tinyone \tc\tinyvarzero \tc\tinyone \tc\tinyvarzero \tc\tinyone \tc\tinyvarzero \tc\tinyone \tc\tinyvarzero \tc\tinyone \tc\tinyvarzero \tc\tinyone \tc\tinyvarzero \tc\tinyone \tc\tinyvarzero \tc\tinyone \tc\tinyvarzero \tc\tinyone \tc\tinyvarzero \tc\tinyone \\
\tc\tinyone \tc\tinyone \tc\tinyvarone \tc\tinyzero \tc\tinyvarzero \tc\tinyzero \tc\tinyvarone \tc\tinyzero \tc\tinyone \tc\tinyone \tc\tinyone \tc\tinyone \tc\tinyvarone \tc\tinyzero \tc\tinyvarzero \tc\tinyzero \tc\tinyvarone \tc\tinyzero \tc\tinyone \tc\tinyone \tc\tinyone \tc\tinyone \tc\tinyvarone \tc\tinyzero \tc\tinyvarzero \tc\tinyzero \tc\tinyone \tc\tinyone \tc\tinyone \tc\tinyzero \tc\tinyzero \tc\tinyzero \tc\tinyone \tc\tinyone \tc\tinyone \tc\tinyzero \tc\zero \tc\zero \tc\one \tc\zero \tc\zero \tc\zero \tc\one \tc\one \tc\zero \tc\zero \tc\zero \tc\zero \tc\one \tc\tinyone \tc\tinyone \tc\tinyzero \tc\tinyone \tc\tinyone \tc\tinyone \tc\tinyzero \nocarry{\tc\tinyone} \tc\tinyvarzero \tc\tinyone \tc\tinyvarzero \tc\tinyone \tc\tinyvarzero \tc\tinyone \tc\tinyvarzero \tc\tinyone \tc\tinyvarzero \tc\tinyone \tc\tinyvarzero \tc\tinyone \tc\tinyvarzero \tc\tinyone \tc\tinyvarzero \tc\tinyone \tc\tinyvarzero \tc\tinyone \tc\tinyvarzero \tc\tinyone \tc\tinyvarzero \tc\tinyone \\
\tc\tinyone \tc\tinyone \tc\tinyvarone \tc\tinyzero \tc\tinyvarzero \tc\tinyzero \tc\tinyvarone \tc\tinyzero \tc\tinyone \tc\tinyone \tc\tinyone \tc\tinyone \tc\tinyvarone \tc\tinyzero \tc\tinyvarzero \tc\tinyzero \tc\tinyvarone \tc\tinyzero \tc\tinyone \tc\tinyone \tc\tinyone \tc\tinyone \tc\tinyvarone \tc\tinyzero \nocarry{\tc\tinyvarzero} \tc\tinyvarone \tc\tinydunno \tc\tinydunno \tc\tinyvarone \tc\tinyvarzero \tc\tinyvarzero \tc\tinyvarone \tc\tinydunno \tc\tinydunno \tc\tinyvarone \tc\varzero \tc\varzero \tc\varone \tc\varone \tc\varzero \tc\varzero \tc\varone \tc\dunno \tc\varone \tc\varzero \tc\varzero \tc\varzero \tc\varone \tc\dunno \tc\dunno \tc\tinyvarone \tc\tinyvarone \tc\tinydunno \tc\tinydunno \tc\tinyvarone \tc\tinyvarone \tc\tinyvarone \tc\tinyvarzero \tc\tinyone \tc\tinyvarzero \tc\tinyone \tc\tinyvarzero \tc\tinyone \tc\tinyvarzero \tc\tinyone \tc\tinyvarzero \tc\tinyone \tc\tinyvarzero \tc\tinyone \tc\tinyvarzero \tc\tinyone \tc\tinyvarzero \tc\tinyone \tc\tinyvarzero \tc\tinyone \tc\tinyvarzero \tc\tinyone \tc\tinyvarzero \tc\tinyone \\
\tc\tinyone \tc\tinyone \tc\tinyvarone \tc\tinyzero \tc\tinyvarzero \tc\tinyzero \tc\tinyvarone \tc\tinyzero \tc\tinyone \tc\tinyone \tc\tinyone \tc\tinyone \tc\tinyvarone \tc\tinyzero \tc\tinyvarzero \tc\tinyzero \tc\tinyvarone \tc\tinyzero \tc\tinyone \tc\tinyone \tc\tinyone \tc\tinyone \tc\tinyvarone \tc\tinyzero \tc\tinyzero \tc\tinyone \tc\tinyone \tc\tinyzero \tc\tinyone \tc\tinyzero \tc\tinyzero \tc\tinyone \tc\tinyone \tc\tinyzero \tc\tinyone \tc\zero \tc\zero \tc\one \tc\one \tc\zero \tc\zero \tc\one \tc\one \tc\one \tc\zero \tc\zero \tc\zero \tc\one \tc\one \tc\zero \tc\tinyone \tc\tinyone \tc\tinyone \tc\tinyzero \tc\tinyone \tc\tinyone \tc\tinyone \tc\tinyzero \nocarry{\tc\tinyone} \tc\tinyvarzero \tc\tinyone \tc\tinyvarzero \tc\tinyone \tc\tinyvarzero \tc\tinyone \tc\tinyvarzero \tc\tinyone \tc\tinyvarzero \tc\tinyone \tc\tinyvarzero \tc\tinyone \tc\tinyvarzero \tc\tinyone \tc\tinyvarzero \tc\tinyone \tc\tinyvarzero \tc\tinyone \tc\tinyvarzero \tc\tinyone \\
\tc\tinyone \tc\tinyone \tc\tinyvarone \tc\tinyzero \tc\tinyvarzero \tc\tinyzero \tc\tinyvarone \tc\tinyzero \tc\tinyone \tc\tinyone \tc\tinyone \tc\tinyone \tc\tinyvarone \tc\tinyzero \tc\tinyvarzero \tc\tinyzero \tc\tinyvarone \tc\tinyzero \tc\tinyone \tc\tinyone \tc\tinyone \tc\tinyone \nocarry{\tc\tinyvarone} \tc\tinyvarzero \tc\tinyvarone \tc\tinydunno \tc\tinyvarone \tc\tinyvarone \tc\tinyvarone \tc\tinyvarzero \tc\tinyvarone \tc\tinydunno \tc\tinyvarone \tc\tinyvarone \tc\varone \tc\varzero \tc\varone \tc\dunno \tc\varone \tc\varzero \tc\varone \tc\dunno \tc\dunno \tc\varone \tc\varzero \tc\varzero \tc\varone \tc\dunno \tc\varone \tc\varone \tc\dunno \tc\tinydunno \tc\tinyvarone \tc\tinyvarone \tc\tinydunno \tc\tinydunno \tc\tinyvarone \tc\tinyvarone \tc\tinyvarone \tc\tinyvarzero \tc\tinyone \tc\tinyvarzero \tc\tinyone \tc\tinyvarzero \tc\tinyone \tc\tinyvarzero \tc\tinyone \tc\tinyvarzero \tc\tinyone \tc\tinyvarzero \tc\tinyone \tc\tinyvarzero \tc\tinyone \tc\tinyvarzero \tc\tinyone \tc\tinyvarzero \tc\tinyone \tc\tinyvarzero \tc\tinyone \\
\tc\tinyone \tc\tinyone \tc\tinyvarone \tc\tinyzero \tc\tinyvarzero \tc\tinyzero \tc\tinyvarone \tc\tinyzero \tc\tinyone \tc\tinyone \tc\tinyone \tc\tinyone \tc\tinyvarone \tc\tinyzero \tc\tinyvarzero \tc\tinyzero \tc\tinyvarone \tc\tinyzero \tc\tinyone \tc\tinyone \tc\tinyone \tc\tinyone \tc\tinyone \tc\tinyzero \tc\tinyone \tc\tinyone \tc\tinyone \tc\tinyone \tc\tinyone \tc\tinyzero \tc\tinyone \tc\tinyone \tc\tinyone \tc\tinyone \tc\one \tc\zero \tc\one \tc\one \tc\one \tc\zero \tc\one \tc\one \tc\zero \tc\one \tc\zero \tc\zero \tc\one \tc\one \tc\one \tc\one \tc\one \tc\tinyzero \tc\tinyone \tc\tinyone \tc\tinyone \tc\tinyzero \tc\tinyone \tc\tinyone \tc\tinyone \tc\tinyzero \nocarry{\tc\tinyone} \tc\tinyvarzero \tc\tinyone \tc\tinyvarzero \tc\tinyone \tc\tinyvarzero \tc\tinyone \tc\tinyvarzero \tc\tinyone \tc\tinyvarzero \tc\tinyone \tc\tinyvarzero \tc\tinyone \tc\tinyvarzero \tc\tinyone \tc\tinyvarzero \tc\tinyone \tc\tinyvarzero \tc\tinyone \\
\tc\tinyone \tc\tinyone \tc\tinyvarone \tc\tinyzero \tc\tinyvarzero \tc\tinyzero \tc\tinyvarone \tc\tinyzero \tc\tinyone \tc\tinyone \tc\tinyone \tc\tinyone \tc\tinyvarone \tc\tinyzero \tc\tinyvarzero \tc\tinyzero \nocarry{\tc\tinyvarone} \tc\tinyvarone \tc\tinydunno \tc\tinydunno \tc\tinydunno \tc\tinydunno \tc\tinyvarone \tc\tinyvarone \tc\tinydunno \tc\tinydunno \tc\tinydunno \tc\tinydunno \tc\tinyvarone \tc\tinyvarone \tc\tinydunno \tc\tinydunno \tc\tinydunno \tc\dunno \tc\varone \tc\varone \tc\dunno \tc\dunno \tc\varone \tc\varone \tc\dunno \tc\varone \tc\varone \tc\varone \tc\varzero \tc\varone \tc\dunno \tc\dunno \tc\dunno \tc\dunno \tc\varone \tc\varone \tc\tinydunno \tc\tinydunno \tc\tinyvarone \tc\tinyvarone \tc\tinydunno \tc\tinydunno \tc\tinyvarone \tc\tinyvarone \tc\tinyvarone \tc\tinyvarzero \tc\tinyone \tc\tinyvarzero \tc\tinyone \tc\tinyvarzero \tc\tinyone \tc\tinyvarzero \tc\tinyone \tc\tinyvarzero \tc\tinyone \tc\tinyvarzero \tc\tinyone \tc\tinyvarzero \tc\tinyone \tc\tinyvarzero \tc\tinyone \tc\tinyvarzero \tc\tinyone \\
\tc\tinyone \tc\tinyone \tc\tinyvarone \tc\tinyzero \tc\tinyvarzero \tc\tinyzero \tc\tinyvarone \tc\tinyzero \tc\tinyone \tc\tinyone \tc\tinyone \tc\tinyone \tc\tinyvarone \tc\tinyzero \tc\tinyvarzero \tc\tinyzero \tc\tinyone \tc\tinyone \tc\tinyone \tc\tinyzero \tc\tinyzero \tc\tinyzero \tc\tinyone \tc\tinyone \tc\tinyone \tc\tinyzero \tc\tinyzero \tc\tinyzero \tc\tinyone \tc\tinyone \tc\tinyone \tc\tinyzero \tc\tinyzero \tc\zero \tc\one \tc\one \tc\one \tc\zero \tc\one \tc\one \tc\one \tc\one \tc\one \tc\one \tc\zero \tc\one \tc\one \tc\zero \tc\zero \tc\zero \tc\one \tc\one \tc\tinyone \tc\tinyzero \tc\tinyone \tc\tinyone \tc\tinyone \tc\tinyzero \tc\tinyone \tc\tinyone \tc\tinyone \tc\tinyzero \nocarry{\tc\tinyone} \tc\tinyvarzero \tc\tinyone \tc\tinyvarzero \tc\tinyone \tc\tinyvarzero \tc\tinyone \tc\tinyvarzero \tc\tinyone \tc\tinyvarzero \tc\tinyone \tc\tinyvarzero \tc\tinyone \tc\tinyvarzero \tc\tinyone \tc\tinyvarzero \tc\tinyone \\
\tc\tinyone \tc\tinyone \tc\tinyvarone \tc\tinyzero \tc\tinyvarzero \tc\tinyzero \tc\tinyvarone \tc\tinyzero \tc\tinyone \tc\tinyone \tc\tinyone \tc\tinyone \tc\tinyvarone \tc\tinyzero \nocarry{\tc\tinyvarzero} \tc\tinyvarone \tc\tinydunno \tc\tinydunno \tc\tinyvarone \tc\tinyvarzero \tc\tinyvarzero \tc\tinyvarone \tc\tinydunno \tc\tinydunno \tc\tinyvarone \tc\tinyvarzero \tc\tinyvarzero \tc\tinyvarone \tc\tinydunno \tc\tinydunno \tc\tinyvarone \tc\tinyvarzero \tc\varzero \tc\varone \tc\dunno \tc\dunno \tc\varone \tc\varone \tc\dunno \tc\dunno \tc\dunno \tc\dunno \tc\dunno \tc\varone \tc\varone \tc\dunno \tc\varone \tc\varzero \tc\varzero \tc\varone \tc\dunno \tc\dunno \tc\varone \tc\tinyvarone \tc\tinydunno \tc\tinydunno \tc\tinyvarone \tc\tinyvarone \tc\tinydunno \tc\tinydunno \tc\tinyvarone \tc\tinyvarone \tc\tinyvarone \tc\tinyvarzero \tc\tinyone \tc\tinyvarzero \tc\tinyone \tc\tinyvarzero \tc\tinyone \tc\tinyvarzero \tc\tinyone \tc\tinyvarzero \tc\tinyone \tc\tinyvarzero \tc\tinyone \tc\tinyvarzero \tc\tinyone \tc\tinyvarzero \tc\tinyone \\
\tc\tinyone \tc\tinyone \tc\tinyvarone \tc\tinyzero \tc\tinyvarzero \tc\tinyzero \tc\tinyvarone \tc\tinyzero \tc\tinyone \tc\tinyone \tc\tinyone \tc\tinyone \tc\tinyvarone \tc\tinyzero \tc\tinyzero \tc\tinyone \tc\tinyone \tc\tinyzero \tc\tinyone \tc\tinyzero \tc\tinyzero \tc\tinyone \tc\tinyone \tc\tinyzero \tc\tinyone \tc\tinyzero \tc\tinyzero \tc\tinyone \tc\tinyone \tc\tinyzero \tc\tinyone \tc\tinyzero \tc\zero \tc\one \tc\one \tc\zero \tc\one \tc\one \tc\one \tc\zero \tc\zero \tc\zero \tc\zero \tc\one \tc\one \tc\one \tc\one \tc\zero \tc\zero \tc\one \tc\one \tc\zero \tc\one \tc\tinyone \tc\tinyone \tc\tinyzero \tc\tinyone \tc\tinyone \tc\tinyone \tc\tinyzero \tc\tinyone \tc\tinyone \tc\tinyone \tc\tinyzero \nocarry{\tc\tinyone} \tc\tinyvarzero \tc\tinyone \tc\tinyvarzero \tc\tinyone \tc\tinyvarzero \tc\tinyone \tc\tinyvarzero \tc\tinyone \tc\tinyvarzero \tc\tinyone \tc\tinyvarzero \tc\tinyone \tc\tinyvarzero \tc\tinyone \\
\tc\tinyone \tc\tinyone \tc\tinyvarone \tc\tinyzero \tc\tinyvarzero \tc\tinyzero \tc\tinyvarone \tc\tinyzero \tc\tinyone \tc\tinyone \tc\tinyone \tc\tinyone \nocarry{\tc\tinyvarone} \tc\tinyvarzero \tc\tinyvarone \tc\tinydunno \tc\tinyvarone \tc\tinyvarone \tc\tinyvarone \tc\tinyvarzero \tc\tinyvarone \tc\tinydunno \tc\tinyvarone \tc\tinyvarone \tc\tinyvarone \tc\tinyvarzero \tc\tinyvarone \tc\tinydunno \tc\tinyvarone \tc\tinyvarone \tc\tinyvarone \tc\varzero \tc\varone \tc\dunno \tc\varone \tc\varone \tc\dunno \tc\dunno \tc\varone \tc\varzero \tc\varzero \tc\varzero \tc\varone \tc\dunno \tc\dunno \tc\dunno \tc\varone \tc\varzero \tc\varone \tc\dunno \tc\varone \tc\varone \tc\dunno \tc\dunno \tc\tinyvarone \tc\tinyvarone \tc\tinydunno \tc\tinydunno \tc\tinyvarone \tc\tinyvarone \tc\tinydunno \tc\tinydunno \tc\tinyvarone \tc\tinyvarone \tc\tinyvarone \tc\tinyvarzero \tc\tinyone \tc\tinyvarzero \tc\tinyone \tc\tinyvarzero \tc\tinyone \tc\tinyvarzero \tc\tinyone \tc\tinyvarzero \tc\tinyone \tc\tinyvarzero \tc\tinyone \tc\tinyvarzero \tc\tinyone \\
\tc\tinyone \tc\tinyone \tc\tinyvarone \tc\tinyzero \tc\tinyvarzero \tc\tinyzero \tc\tinyvarone \tc\tinyzero \tc\tinyone \tc\tinyone \tc\tinyone \tc\tinyone \tc\tinyone \tc\tinyzero \tc\tinyone \tc\tinyone \tc\tinyone \tc\tinyone \tc\tinyone \tc\tinyzero \tc\tinyone \tc\tinyone \tc\tinyone \tc\tinyone \tc\tinyone \tc\tinyzero \tc\tinyone \tc\tinyone \tc\tinyone \tc\tinyone \tc\tinyone \tc\zero \tc\one \tc\one \tc\one \tc\one \tc\one \tc\zero \tc\one \tc\zero \tc\zero \tc\zero \tc\one \tc\one \tc\zero \tc\zero \tc\one \tc\zero \tc\one \tc\one \tc\one \tc\one \tc\one \tc\zero \tc\tinyone \tc\tinyone \tc\tinyone \tc\tinyzero \tc\tinyone \tc\tinyone \tc\tinyone \tc\tinyzero \tc\tinyone \tc\tinyone \tc\tinyone \tc\tinyzero \nocarry{\tc\tinyone} \tc\tinyvarzero \tc\tinyone \tc\tinyvarzero \tc\tinyone \tc\tinyvarzero \tc\tinyone \tc\tinyvarzero \tc\tinyone \tc\tinyvarzero \tc\tinyone \tc\tinyvarzero \tc\tinyone \\
\tc\tinyone \tc\tinyone \tc\tinyvarone \tc\tinyzero \tc\tinyvarzero \tc\tinyzero \nocarry{\tc\tinyvarone} \tc\tinyvarone \tc\tinydunno \tc\tinydunno \tc\tinydunno \tc\tinydunno \tc\tinyvarone \tc\tinyvarone \tc\tinydunno \tc\tinydunno \tc\tinydunno \tc\tinydunno \tc\tinyvarone \tc\tinyvarone \tc\tinydunno \tc\tinydunno \tc\tinydunno \tc\tinydunno \tc\tinyvarone \tc\tinyvarone \tc\tinydunno \tc\tinydunno \tc\tinydunno \tc\tinydunno \tc\varone \tc\varone \tc\dunno \tc\dunno \tc\dunno \tc\dunno \tc\varone \tc\varone \tc\varone \tc\varzero \tc\varzero \tc\varone \tc\dunno \tc\varone \tc\varzero \tc\varone \tc\varone \tc\varone \tc\dunno \tc\dunno \tc\dunno \tc\dunno \tc\varone \tc\varone \tc\dunno \tc\tinydunno \tc\tinyvarone \tc\tinyvarone \tc\tinydunno \tc\tinydunno \tc\tinyvarone \tc\tinyvarone \tc\tinydunno \tc\tinydunno \tc\tinyvarone \tc\tinyvarone \tc\tinyvarone \tc\tinyvarzero \tc\tinyone \tc\tinyvarzero \tc\tinyone \tc\tinyvarzero \tc\tinyone \tc\tinyvarzero \tc\tinyone \tc\tinyvarzero \tc\tinyone \tc\tinyvarzero \tc\tinyone \\
\tc\tinyone \tc\tinyone \tc\tinyvarone \tc\tinyzero \tc\tinyvarzero \tc\tinyzero \tc\tinyone \tc\tinyone \tc\tinyone \tc\tinyzero \tc\tinyzero \tc\tinyzero \tc\tinyone \tc\tinyone \tc\tinyone \tc\tinyzero \tc\tinyzero \tc\tinyzero \tc\tinyone \tc\tinyone \tc\tinyone \tc\tinyzero \tc\tinyzero \tc\tinyzero \tc\tinyone \tc\tinyone \tc\tinyone \tc\tinyzero \tc\tinyzero \tc\tinyzero \tc\one \tc\one \tc\one \tc\zero \tc\zero \tc\zero \tc\one \tc\one \tc\one \tc\zero \tc\zero \tc\one \tc\one \tc\one \tc\zero \tc\one \tc\one \tc\one \tc\one \tc\zero \tc\zero \tc\zero \tc\one \tc\one \tc\one \tc\tinyzero \tc\tinyone \tc\tinyone \tc\tinyone \tc\tinyzero \tc\tinyone \tc\tinyone \tc\tinyone \tc\tinyzero \tc\tinyone \tc\tinyone \tc\tinyone \tc\tinyzero \nocarry{\tc\tinyone} \tc\tinyvarzero \tc\tinyone \tc\tinyvarzero \tc\tinyone \tc\tinyvarzero \tc\tinyone \tc\tinyvarzero \tc\tinyone \tc\tinyvarzero \tc\tinyone \\
\tc\tinyone \tc\tinyone \tc\tinyvarone \tc\tinyzero \nocarry{\tc\tinyvarzero} \tc\tinyvarone \tc\tinydunno \tc\tinydunno \tc\tinyvarone \tc\tinyvarzero \tc\tinyvarzero \tc\tinyvarone \tc\tinydunno \tc\tinydunno \tc\tinyvarone \tc\tinyvarzero \tc\tinyvarzero \tc\tinyvarone \tc\tinydunno \tc\tinydunno \tc\tinyvarone \tc\tinyvarzero \tc\tinyvarzero \tc\tinyvarone \tc\tinydunno \tc\tinydunno \tc\tinyvarone \tc\tinyvarzero \tc\tinyvarzero \tc\varone \tc\dunno \tc\dunno \tc\varone \tc\varzero \tc\varzero \tc\varone \tc\dunno \tc\dunno \tc\varone \tc\varzero \tc\varone \tc\dunno \tc\dunno \tc\varone \tc\varone \tc\dunno \tc\dunno \tc\dunno \tc\varone \tc\varzero \tc\varzero \tc\varone \tc\dunno \tc\dunno \tc\varone \tc\varone \tc\tinydunno \tc\tinydunno \tc\tinyvarone \tc\tinyvarone \tc\tinydunno \tc\tinydunno \tc\tinyvarone \tc\tinyvarone \tc\tinydunno \tc\tinydunno \tc\tinyvarone \tc\tinyvarone \tc\tinyvarone \tc\tinyvarzero \tc\tinyone \tc\tinyvarzero \tc\tinyone \tc\tinyvarzero \tc\tinyone \tc\tinyvarzero \tc\tinyone \tc\tinyvarzero \tc\tinyone \\
\tc\tinyone \tc\tinyone \tc\tinyvarone \tc\tinyzero \tc\tinyzero \tc\tinyone \tc\tinyone \tc\tinyzero \tc\tinyone \tc\tinyzero \tc\tinyzero \tc\tinyone \tc\tinyone \tc\tinyzero \tc\tinyone \tc\tinyzero \tc\tinyzero \tc\tinyone \tc\tinyone \tc\tinyzero \tc\tinyone \tc\tinyzero \tc\tinyzero \tc\tinyone \tc\tinyone \tc\tinyzero \tc\tinyone \tc\tinyzero \tc\tinyzero \tc\one \tc\one \tc\zero \tc\one \tc\zero \tc\zero \tc\one \tc\one \tc\zero \tc\one \tc\zero \tc\one \tc\one \tc\zero \tc\one \tc\one \tc\one \tc\zero \tc\zero \tc\one \tc\zero \tc\zero \tc\one \tc\one \tc\zero \tc\one \tc\one \tc\tinyone \tc\tinyzero \tc\tinyone \tc\tinyone \tc\tinyone \tc\tinyzero \tc\tinyone \tc\tinyone \tc\tinyone \tc\tinyzero \tc\tinyone \tc\tinyone \tc\tinyone \tc\tinyzero \nocarry{\tc\tinyone} \tc\tinyvarzero \tc\tinyone \tc\tinyvarzero \tc\tinyone \tc\tinyvarzero \tc\tinyone \tc\tinyvarzero \tc\tinyone \\
\tc\tinyone \tc\tinyone \nocarry{\tc\tinyvarone} \tc\tinyvarzero \tc\tinyvarone \tc\tinydunno \tc\tinyvarone \tc\tinyvarone \tc\tinyvarone \tc\tinyvarzero \tc\tinyvarone \tc\tinydunno \tc\tinyvarone \tc\tinyvarone \tc\tinyvarone \tc\tinyvarzero \tc\tinyvarone \tc\tinydunno \tc\tinyvarone \tc\tinyvarone \tc\tinyvarone \tc\tinyvarzero \tc\tinyvarone \tc\tinydunno \tc\tinyvarone \tc\tinyvarone \tc\tinyvarone \tc\tinyvarzero \tc\varone \tc\dunno \tc\varone \tc\varone \tc\varone \tc\varzero \tc\varone \tc\dunno \tc\varone \tc\varone \tc\varone \tc\varone \tc\dunno \tc\varone \tc\varone \tc\dunno \tc\dunno \tc\varone \tc\varzero \tc\varone \tc\varone \tc\varzero \tc\varone \tc\dunno \tc\varone \tc\varone \tc\dunno \tc\dunno \tc\varone \tc\tinyvarone \tc\tinydunno \tc\tinydunno \tc\tinyvarone \tc\tinyvarone \tc\tinydunno \tc\tinydunno \tc\tinyvarone \tc\tinyvarone \tc\tinydunno \tc\tinydunno \tc\tinyvarone \tc\tinyvarone \tc\tinyvarone \tc\tinyvarzero \tc\tinyone \tc\tinyvarzero \tc\tinyone \tc\tinyvarzero \tc\tinyone \tc\tinyvarzero \tc\tinyone \\
\tc\tinyone \tc\tinyone \tc\tinyone \tc\tinyzero \tc\tinyone \tc\tinyone \tc\tinyone \tc\tinyone \tc\tinyone \tc\tinyzero \tc\tinyone \tc\tinyone \tc\tinyone \tc\tinyone \tc\tinyone \tc\tinyzero \tc\tinyone \tc\tinyone \tc\tinyone \tc\tinyone \tc\tinyone \tc\tinyzero \tc\tinyone \tc\tinyone \tc\tinyone \tc\tinyone \tc\tinyone \tc\tinyzero \tc\one \tc\one \tc\one \tc\one \tc\one \tc\zero \tc\one \tc\one \tc\one \tc\one \tc\one \tc\one \tc\one \tc\one \tc\one \tc\one \tc\zero \tc\one \tc\zero \tc\one \tc\one \tc\zero \tc\one \tc\one \tc\one \tc\one \tc\one \tc\zero \tc\one \tc\tinyone \tc\tinyone \tc\tinyzero \tc\tinyone \tc\tinyone \tc\tinyone \tc\tinyzero \tc\tinyone \tc\tinyone \tc\tinyone \tc\tinyzero \tc\tinyone \tc\tinyone \tc\tinyone \tc\tinyzero \nocarry{\tc\tinyone} \tc\tinyvarzero \tc\tinyone \tc\tinyvarzero \tc\tinyone \tc\tinyvarzero \tc\tinyone \\
\cellsend$\end{center}
In Wolfram and Cook's proof of the universality of $R$, the state of
the emulated universal cyclic tag system $U$ is recovered from the
causal future of the input $I$ alone (possible since the update rule
observes only nearest neighbours).
Thus our emulation of $R$ on the future of $I$ suffices for universality of
the Turing machine.

In overview, the wrap construction on $X$ and $Y$ will proceed as
follows.
Consider the case $X=111011$ as in the example above.  Run the
\emph{wrapped} form of the cellular automaton on $X$,
that is, with just
six
cells (the length of $X$), and where the last cell is formally
considered to be the left neighbour of the first.
Truncate the computation just as a row is about to
recur (in this case, the first row).
The result is below-left:
\begin{center}
\renewcommand{\cellw}{3ex}
\vspace*{.5ex}
$\cells
\cell 1 \cell 1 \cell 1 \cell 0 \cell 1 \cell 1 \\
\cell 0 \cell 0 \cell 1 \cell 1 \cell 1 \cell 0 \\
\cell 0 \cell 1 \cell 1 \cell 0 \cell 1 \cell 0 \\
\cell 1 \cell 1 \cell 1 \cell 1 \cell 1 \cell 0 \\
\cell 1 \cell 0 \cell 0 \cell 0 \cell 1 \cell 1 \\
\cell 1 \cell 0 \cell 0 \cell 1 \cell 1 \cell 0 \\
\cell 1 \cell 0 \cell 1 \cell 1 \cell 1 \cell 1 \\
\cell 1 \cell 1 \cell 1 \cell 0 \cell 0 \cell 0 \\
\cell 1 \cell 0 \cell 1 \cell 0 \cell 0 \cell 1
\cellsend$
\hspace*{20ex}%
\newcommand{\m}{\begin{picture}(0,0)\put(0,0){\qbezier[4](1,7)(7,10)(13,14)}\end{picture}}%
\newcommand{\p}{\begin{picture}(0,0)\put(1,1){\qbezier[2](1,3)(4,3)(7,3)}\end{picture}}%
\newcommand{\q}{\begin{picture}(0,0)\put(0,1){\qbezier[2](0,3)(3,3)(6,3)}\end{picture}}%
\newcommand{\n}{\begin{picture}(0,0)\put(-2,2){\vector(-2,-1){20}}\end{picture}}%
\newcommand{\x}{\begin{picture}(0,0)\put(1,4){\vector(-1,0){10}}\end{picture}}%
\newcommand{\xx}{\begin{picture}(0,0)\put(1,4){\vector(-1,0){8}}\end{picture}}%
\newcommand{\I}{\mbox{\x\large$\mathbf{1}$}}%
\newcommand{\J}{\mbox{\xx\large$\mathbf{1}$}}%
\newcommand{\Q}{\mbox{\n\large$\mathbf{0}$}}%
\newcommand{\Qi}{\mbox{\n\large$\mathbf{0}$\m}}%
\newcommand{\Qj}{\mbox{\n\large$\mathbf{0}$\p}}%
\newcommand{\Ij}{\mbox{\x\large$\mathbf{1}$\q}}%
\newcommand{\Ii}{\mbox{\x\large$\mathbf{1}$\m}}%
$\cells
\cell\I \cell\Ii\cell 1 \cell\Q \cell\J \cell\Ij \\
\cell 0 \cell\Q \cell 1 \cell 1 \cell 1 \cell 0 \\
\cell 0 \cell 1 \cell 1 \cell 0 \cell 1 \cell\Qi \\
\cell\I \cell\I \cell\I \cell\I \cell 1 \cell\Qj \\
\cell 1 \cell 0 \cell 0 \cell\Q \cell 1 \cell 1 \\
\cell 1 \cell\Q \cell 0 \cell 1 \cell 1 \cell 0 \\
\cell 1 \cell\Q \cell\J \cell\I \cell\I \cell\Ii \\
\cell 1 \cell 1 \cell 1 \cell 0 \cell 0 \cell\Qi \\
\cell 1 \cell 0 \cell 1 \cell\Q \cell 0 \cell 1
\cellsend$
\vspace*{.5ex}
\end{center}
Place a cursor on the top-right cell of the matrix, then build a word $W$ by
repeating:\label{left-wrap-informal}
\begin{itemize}
\item[$\star$] Write down the symbol $s$ at the cursor.
\begin{itemize}
\item If $s=1$, move the cursor to the cell to the left. Go to $\star$.
\item If $s=0$, let $t$ be the symbol to the left of the cursor.  Write down $\underline t$.  
Move the cursor to the cell which is two columns to the left and one
row down.  Go to $\star$.
\end{itemize}
\end{itemize}
The $9\times 6$ matrix is \emph{wrapped}: moving the cursor down from
the bottom row takes us to the top row\footnote{More generally, moving
down from the bottom row takes us to the row which was about to recur.
For this $X$, that row was the first row.  See Section~\ref{details}
for details.}  (column unchanged), and moving the cursor left from
the first column takes us to the last column (row unchanged).
The figure above-right shows in bold every cell which is visited, with
cursor moves as arrows.
We terminate when the cursor lands on a cell which has already been
visited (in this case, the top-right $\mathbf{1}$ cell where we started).  
Every time we write a new symbol, we add it to the \emph{left} of $W$.
The resulting 30-symbol word $W$ is:
$$
11\varone0\varzero0\varone011
\;\;11\varone0\varzero0\varone011
\;\;11\varone0\varzero0\varone011
$$ 
(The gaps merely emphasise repetition in $W$.) We define $\leftseed{X}$, 
the \defn{left seed}, as the shortest word which yields $W$ by
repetition:
$$\leftseed  X \;\;=\;\; 11\varone 0 \varzero 0\varone 011$$
The reader can verify that this is indeed the word seeding the
tape to the left of the input $I=10011$ in the successful emulation above.  

The wrap construction yields $\rightseed{Y}$ from $Y$ in a similar
manner.  Consider $Y=1101$, again from the example emulation above.
Running the wrapped form of the cellular automaton on $Y$ yields just
two rows (since $1101$ recurs after only the second wrapped application of the
rule), as below-left:
\begin{center}\renewcommand{\cellw}{3ex}
\vspace*{.5ex}
$\cells
\cell 1 \cell 1 \cell 0 \cell 1 \\
\cell 0 \cell 1 \cell 1 \cell 1
\cellsend$
\hspace*{24ex}
\newcommand{\m}{\begin{picture}(0,0)\put(-1,0){\qbezier[4](-1,7)(-7,10)(-13,13)}\end{picture}}%
\newcommand{\n}{\begin{picture}(0,0)\put(1,2){\vector(2,-1){20}}\end{picture}}%
\newcommand{\x}{\begin{picture}(0,0)\put(-1,4){\vector(1,0){10}}\end{picture}}%
\newcommand{\y}{\begin{picture}(0,0)\put(-1,4){\vector(1,0){8}}\end{picture}}%
\newcommand{\Q}{\mbox{\large\m$\mathbf{0}$\n}}%
\newcommand{\I}{\mbox{\large$\mathbf{1}$\x}}%
\newcommand{\J}{\mbox{\large$\mathbf{1}$\y}}%
$\cells
\cell\I \cell\J \cell\Q \cell 1 \\
\cell\Q \cell 1 \cell 1 \cell 1
\cellsend$
\vspace*{.5ex}
\end{center}
Place a cursor on the top-left cell, then build a word $W$ by
repeating:\label{right-wrap-informal}
\begin{itemize}
\item[$\star$] Let $s$ be the symbol at the cursor.  Write down $\underline s$ (the underlined variant of $s$).
\begin{itemize}
\item If $s=1$, move the cursor to the cell to the right. Go to $\star$.
\item If $s=0$, move the cursor to the cell to the right. Go to $\star\star$.
\end{itemize}
\item[$\star\star$] Let $s$ be the symbol at the cursor.  Write down $\underline s$ (the underlined variant of $s$).
\begin{itemize}
\item If $s=0$, write down $\underline 0$ and move the cursor to the
  cell to the right. Go to $\star\star$.
\item If $s=1$, write down $1$ and move the cursor to the cell which
  is one column to the right and one row down.  Go to $\star$.
\end{itemize}
\end{itemize}
As before, we terminate when the cursor lands on a cell which has
already been visited (in this case, the third cell, $\mathbf{0}$, in the top
row).
Every time we write down a new symbol, we add it to the \emph{right}
of $W$.  
The resulting 6-symbol word $W$ is:
$$\varone\,\varone\,\varzero\,1\,\varzero\,1$$
The first two symbols $\varone\varone$ in $W$ came from the two $1$'s
at the beginning of the top row in the $2\times 4$ matrix.  These two $1$'s are
not part of the cycle of the cursor: were we to continue making cursor
moves, we would never revisit them.  This part of $W$
becomes $\rightstem Y$, called the \defn{right stem}:
$$\rightstem Y\;\;=\;\;\varone\,\varone$$
The shortest word whose repetition yields the remainder
$\varzero\,1\,\varzero\,1$ of $W$ becomes $\rightseed Y$, the \defn{right
seed}:\footnote{The \defn{left stem} $\leftstem X$ was empty (since the
cursor returned to the top-right cell of the $9\times 6$ grid, on
which it started), therefore we simplified the exposition by not
mentioning it.  In general, however, there may be a non-empty left
stem.}
$$\rightseed Y\;\;=\;\;\varzero\,1$$
Observe that the initial tape to the right of the input $I$ in the
emulation above (p.\,\pageref{emulation}) is the right stem $\rightstem Y$ followed by the
infinite repetition of the right seed $\rightseed Y$.

The construction described here does not work for words $X$ and $Y$
with the property that one of the rows in our constructed matrix
consists of 0s only. For those words a simpler construction is
possible. We do not describe this simpler construction here, as the
words $X$ and $Y$ used by Wolfram and Cook in the simulation of $U$ by
$R$ have the property that such a row of 0s will not emerge.

\section{Formal details}\label{details}\vspace*{1ex}

Let $\N=\{0,1,\ldots\}$ and $\Z=\{\ldots,-1,0,1,\ldots\}$.

\vspace*{-1ex}\paragraph*{Words.}
Let $\Sigma$ be a set of \defn{symbols}.  A \defn{word over $\Sigma$},
or \defn{$\Sigma$-word}, of length $k\in\N$ is a
function $w:\{0,\ldots,k-1\}\to \Sigma$.
We shall often write the sequence $w(0)w(1)\cdots w(k-1)$ of outputs
of $w$, in order, to denote $w$.
For example, $1011$ denotes the $\{0,1\}$-word $w$ of length 4
with $w(0)=w(2)=w(3)=1$ and $w(1)=0$.
A length 0 word is \defn{empty}.
Define the \defn{reverse} $\rev w$ of $w$, also of length $k$, by
$\rev w(i)=w(k-1-i)$.
For example $\rev{(11010)}=01011$.
A $\Sigma$-word $w$ of length $k$ \defn{contains} $x\in\Sigma$ if $w(i)=x$ for some
$i\in\{0,\ldots,k-1\}$.

Let $v$ and $w$ be $\Sigma$-words of length $k$ and $l$, respectively.
The \defn{concatenation} $v\concat w$ of $v$ and $w$ is the
$\Sigma$-word of length $k+l$ defined by $(v\concat w)(i)=v(i)$ for
$0\le i<k$ and $(v\concat w)(k+j)=w(j)$ for $0\le j<l$.
For example, if $v=1011$ and $w=00001$ then $v\concat
w\,=\,101100001$.  Since $u\concat(v\concat w)=(u\concat v)\concat
w$ we can write $u\concat v\concat w$ without ambiguity.
For a $\Sigma$-word $w$ of length $k$ and $0\le a<k$ define $w_{<a}$
as the restriction of $w$ to $\{0,\ldots,a-1\}$ (a $\Sigma$-word of
length $a$) and define $w_{\ge a}$ as the remainder: the unique
$\Sigma$-word $v$ such that $w=w_{<a}\concat v$ (a $\Sigma$-word of
length $k-a$).

Write $\Sigma^*$ for the set of $\Sigma$-words.
Given a $\Sigma^*$-word $W$ of length $k$ (a ``word of words'') 
its \defn{flattening} $W^\concat$ is the $\Sigma$-word $W(0)\concat
W(1)\concat\cdots\concat W(k-1)$.
For example, if $W\;=\;(111)(00)(1101)$ (the $\{0,1\}^*$-word of
length $3$ given by $W(0)=111$, $W(1)=00$ and $W(2)=1101$) then
$W^\concat\,=\,111001101$.
The \defn{$n$-fold repetition} $w^n$ of $w$ is empty if $n=0$ and
$w\concat w^{n-1}$ if $n>0$ (\ie, $w\concat w\concat\cdots\concat w$ with $n$ occurrences of $w$).
The \defn{reduction} $\red{w}$ of $w$, when it exists, is the shortest
(minimal length) word $r$ such that $w=r^n$ for some $n\ge 1$.
For example, $|110110110110|=110$.

\paragraph*{State.} A \defn{state over $\Sigma$}, or \defn{$\Sigma$-state},
is a function $S\,:\,\Z\,\to\,\Sigma$.  Each $c$ in the domain
$\Z$ of $S$ is a \defn{cell}.  

Let $A,I,B$ be a $\Sigma$-words of lengths $a,l,b$ respectively.
Define $$\cells \cell{\repleft A}\;\, \cell I \,\; \cell{\repright B}
\cellsend$$ as the $\Sigma$-state $S:\Z\to\Sigma$ comprising $I$ on cells $0$ to $l-1$,
infinite repetitions of $A$ to the left, and infinite repetitions of
$B$ to the right: $$S(c)=
\begin{cases}
\,I(c)
     & \text{if }0\le c< l \\
\,A(c\remainder a)
     & \text{if }c < 0 \\
\,B((c-l)\remainder b)
     & \text{if }c \ge l
\end{cases}$$
Given additional $\Sigma$-words $P$ and $Q$ of lengths $p$ and $q$,
respectively, define 
$$\cells \cell{\repleft A}\;\;\,\cell{P}\, \cell I\, \cell{Q} \,\;\; \cell{\repright B}
\cellsend$$
as the $\Sigma$-state $T:\Z\to\Sigma$ obtained by inserting $P$ and
$Q$ either side of $I$ in $\repleft A\: I\:\repright B\:$: 
$$T(c)=
\begin{cases}
\,I(c)
     & \text{if }0\le c<l  \\
\,A((c+p)\remainder a)
     & \text{if }c < -p \\
\,B((c-l-q)\remainder m)
     & \text{if }c \ge l+q\\
\,P(c+p)
     & \text{if }-p \le c < 0\\
\,Q(c-l)
     & \text{if }l \le c < l+q
\end{cases}$$

\subsection{The rule 110 cellular automaton $R$}
Define
\defn{spacetime} as the product $\N\times \Z$ where
$\N$ is the set of \defn{times} and
$\Z$ is the set of \defn{cells} (space).
Each ordered pair $(t,c)\in\N\times \Z$ is an \defn{event}
(spacetime coordinate).
A \defn{run} of $R$ is a function $\rho:\N\times\Z\to\{0,1\}$ such
that for all times $t\in \N$ and cells $c\in \Z$ the following condition holds:
\begin{itemize}
\item \defn{Causality}.
$\rho(t+1,c)\neq \rho(t,c)$ if and only if:
\begin{itemize}
\item (\textsl{Birth}\footnote{``A $1$ is born from a $0$ whose right neighbour is $1$,''
 visually \mbox{\scriptsize$\begin{array}{c@{}c}0&1\\[-.8ex]1\end{array}$}.}) 
$\,\;\;\rho(t,c)=0$ and $\rho(t,c+1)=1\;$, or
\item (\textsl{Death}\footnote{``A $1$ dies by overcrowding when both neighbours are $1$,'' 
visually \mbox{\scriptsize$\begin{array}{c@{}c@{}c}1&1&1\\[-.8ex]&0\end{array}$}.}) 
$\;\rho(t,c-1)=\rho(t,c)=\rho(t,c+1)=1\,$.
\end{itemize}
\end{itemize}
For $t\in\N$ the \defn{$t\nth$ state} or \defn{state at time $t$} of
a run $\rho$, denoted $\rho_t$, is the $\{0,1\}$-state
$\rho_t:\Z\to\{0,1\}$ given by $\rho_t(c)=\rho(t,c)$ for all cells
$c\in\Z$.
The \defn{initial state} of $\rho$ is $\rho_0$.
Note that $R$ is deterministic: the state $\rho_t$ at each time $t\in\N$ is
determined by the initial state $\rho_0$.

\paragraph*{Wrapped rule 110.}

For $n\in\N$ define \defn{$n$-wrapped spacetime}
as $\N\times \{0,\ldots,n-1\}$.
An \defn{$n$-wrapped run}
of the cellular automaton
$R$ is a function $\rho:\N\times \{0,\ldots,n-1\}\to \{0,1\}$ which,
for all times $t\in\N$ and cells $c\in\{0,\ldots,n-1\}$, satisfies
the \textit{Causality} condition defined above upon interpreting $\rho(t,n)$ as $\rho(t,0)$ and $\rho(t,-1)$ as $\rho(t,n-1)$.
A $6$-wrapped run $\rho$ is depicted below-left.
\begin{center}
\renewcommand{\cellw}{3ex}
\vspace*{.5ex}
$\cells
 \cell 0 \cell 0 \cell 0 \cell 1 \cell 0 \cell 1 \\
 \cell 0 \cell 0 \cell 1 \cell 1 \cell 1 \cell 1 \\
 \cell 0 \cell 1 \cell 1 \cell 0 \cell 0 \cell 1 \\
\cell 1 \cell 1 \cell 1 \cell 0 \cell 1 \cell 1 \\
\cell 0 \cell 0 \cell 1 \cell 1 \cell 1 \cell 0 \\
\cell 0 \cell 1 \cell 1 \cell 0 \cell 1 \cell 0 \\
\cell 1 \cell 1 \cell 1 \cell 1 \cell 1 \cell 0 \\
\cell 1 \cell 0 \cell 0 \cell 0 \cell 1 \cell 1 \\
\cell 1 \cell 0 \cell 0 \cell 1 \cell 1 \cell 0 \\
\cell 1 \cell 0 \cell 1 \cell 1 \cell 1 \cell 1 \\
\cell 1 \cell 1 \cell 1 \cell 0 \cell 0 \cell 0 \\
\cell 1 \cell 0 \cell 1 \cell 0 \cell 0 \cell 1 \\
\cell 1 \cell 1 \cell 1 \cell 0 \cell 1 \cell 1 \\
\cell 0 \cell 0 \cell 1 \cell 1 \cell 1 \cell 0 \\
\cell 0 \cell 1 \cell 1 \cell 0 \cell 1 \cell 0 \\
\cell 1 \cell 1 \cell 1 \cell 1 \cell 1 \cell 0 \\
\cell{\vdots}
\cellsend$
\hspace*{8ex}
$\cells
 \cell 0 \cell 0 \cell 0 \cell 1 \cell 0 \cell 1 \\
 \cell 0 \cell 0 \cell 1 \cell 1 \cell 1 \cell 1 \\
 \cell 0 \cell 1 \cell 1 \cell 0 \cell 0 \cell 1 \\[-.1ex] \hline
\cell 1 \cell 1 \cell 1 \cell 0 \cell 1 \cell 1 \\
\cell 0 \cell 0 \cell 1 \cell 1 \cell 1 \cell 0 \\
\cell 0 \cell 1 \cell 1 \cell 0 \cell 1 \cell 0 \\
\cell 1 \cell 1 \cell 1 \cell 1 \cell 1 \cell 0 \\
\cell 1 \cell 0 \cell 0 \cell 0 \cell 1 \cell 1 \\
\cell 1 \cell 0 \cell 0 \cell 1 \cell 1 \cell 0 \\
\cell 1 \cell 0 \cell 1 \cell 1 \cell 1 \cell 1 \\
\cell 1 \cell 1 \cell 1 \cell 0 \cell 0 \cell 0 \\
\cell 1 \cell 0 \cell 1 \cell 0 \cell 0 \cell 1 \\ 
\cell{} \\ \cell{} \\ \cell{} \\ \cell{} \\ \cell{} \\[1.2ex]
\cellsend$
\hspace*{8ex}
\newcommand{\m}{\begin{picture}(0,0)\put(0,0){\qbezier[4](1,7)(7,10)(13,14)}\end{picture}}%
\newcommand{\p}{\begin{picture}(0,0)\put(1,1){\qbezier[2](1,3)(4,3)(7,3)}\end{picture}}%
\newcommand{\q}{\begin{picture}(0,0)\put(0,1){\qbezier[2](0,3)(3,3)(6,3)}\end{picture}}%
\newcommand{\n}{\begin{picture}(0,0)\put(-2,2){\vector(-2,-1){20}}\end{picture}}%
\newcommand{\x}{\begin{picture}(0,0)\put(1,4){\vector(-1,0){10}}\end{picture}}%
\newcommand{\xx}{\begin{picture}(0,0)\put(1,4){\vector(-1,0){8}}\end{picture}}%
\newcommand{\I}{\mbox{\x\large$\mathbf{1}$}}%
\newcommand{\J}{\mbox{\xx\large$\mathbf{1}$}}%
\newcommand{\Q}{\mbox{\n\large$\mathbf{0}$}}%
\newcommand{\Qi}{\mbox{\n\large$\mathbf{0}$\m}}%
\newcommand{\Qj}{\mbox{\n\large$\mathbf{0}$\p}}%
\newcommand{\Ij}{\mbox{\x\large$\mathbf{1}$\q}}%
\newcommand{\Ii}{\mbox{\x\large$\mathbf{1}$\m}}%
\newcommand{\Ji}{\mbox{\xx\large$\mathbf{1}$\m}}%
$\cells
 \cell 0 \cell 0 \cell 0 \cell 1 \cell\Q \cell\J \\
 \cell 0 \cell\Q \cell\J \cell 1 \cell 1 \cell 1 \\
 \cell 0 \cell 1 \cell 1 \cell 0 \cell\Q \cell\Ji \\[-.1ex] \hline
\cell\I \cell\Ii\cell\I \cell\Q \cell\J \cell\Ij \\
\cell 0 \cell\Q \cell 1 \cell 1 \cell 1 \cell 0 \\
\cell 0 \cell 1 \cell 1 \cell 0 \cell 1 \cell\Qi \\
\cell\I \cell\I \cell\I \cell\I \cell 1 \cell\Qj \\
\cell 1 \cell 0 \cell 0 \cell\Q \cell 1 \cell 1 \\
\cell 1 \cell\Q \cell 0 \cell 1 \cell 1 \cell 0 \\
\cell 1 \cell\Q \cell\J \cell\I \cell\I \cell\Ii \\
\cell 1 \cell 1 \cell 1 \cell 0 \cell 0 \cell\Qi \\
\cell 1 \cell 0 \cell 1 \cell\Q \cell 0 \cell 1 \\[5ex]
\textbf{Figure A}
\\[7.3ex]
\cellsend$
\end{center}
The
\defn{$t\nth$ word} of an $n$-wrapped run $\rho$, denoted $\rho_t$,
is given by $\rho_t(c)=\rho(t,c)$ for all $c\in\{0,\ldots,n-1\}$.
The \defn{initial word} of $\rho$ is $\rho_0$.
Wrapped runs are deterministic: $\rho_t$ is determined at each
time $t$ by the initial word $\rho_0$.
The figure above-left shows words $\rho_0$ to $\rho_{15}$, from top to bottom.

The \defn{onset of periodicity} in $\rho$ is the least time $\alpha\in\N$
such that $\rho_\alpha=\rho_{\alpha+\delta}$ for some
$\delta>0$.\footnote{The onset of periodicity $\alpha$ exists since there are at most
$2^n$ distinct $\{0,1\}$-words of length $n$.}
The least such $\delta$ is the \defn{period} of $\rho$, and $\alpha+\delta$
is the \defn{time of first repetition}.
In the example above-left, $\alpha=3$ and $\delta=9$
($\rho_3=\rho_{3+9}=111011$).

The following lemma is trivial, but we state and prove it properly nonetheless.
\begin{lemma}[Periodicity]
Let $\rho$ be an $n$-wrapped run with period $\delta$ and onset of periodicity $\alpha$.
Then
$\rho_{t+\delta}=\rho_t$ for all $t\ge\alpha$.
\end{lemma}
\begin{proof}
By induction on $t$.  Induction base: the condition holds for
$t=\alpha$, by definition of $\alpha$.  Induction step: if
$\rho_t=\rho_u$ then $\rho_{t+1}=\rho_{u+1}$, by the
\textit{Causality} condition defining an $n$-wrapped run.
\end{proof}

\subsection{The wrap constructions}

\paragraph*{The left wrap construction.}
Let $\rho$ be an $n$-wrapped run with period $\delta$ and onset of
periodicity $\alpha$.
Let $\beta=\alpha+\delta$, the time of first repetition.
The \defn{matrix} $\matr{\rho}$ of $\rho$ is the restriction of $\rho$
to times prior to $\beta$, \ie, the restriction of $\rho$ to the domain
$\{0,\ldots,\beta-1\}\times\{0,\ldots,n-1\}$.
For the example $\rho$ depicted above-left, the matrix $\matr{\rho}$
is shown above-centre.  The horizontal rule is a visual aid to
emphasise the onset of periodicity $\alpha=3$ and the period $\delta=9$.

Let $\rho$ be an $n$-wrapped run with period $\delta$ and onset of
periodicity $\alpha$.
Let $\beta=\alpha+\delta$, the time of first repetition.
Write $E$ for the domain $\{0,\ldots,\beta-1\}\times\{0,\ldots,n-1\}$
of the matrix $\matr\rho$. 
A \defn{trajectory} in $\matr\rho$ is a function (infinite sequence) $\N\to E$.
The \defn{left wrap trajectory} of $\rho$ is the trajectory
$\leftwraptraj{\rho}$ in $\matr\rho$ defined by the following
recursion:\footnote{See page~\pageref{left-wrap-informal} for a more
informal presentation.}
\begin{itemize}
\item $\leftwraptraj{\rho}(0)=(0,n-1)\,$\footnote{The top-right cell of the matrix.}
\item if $\leftwraptraj{\rho}(i)=(t,c)$ then $\leftwraptraj{\rho}(i+1)=
\begin{cases}
\big(\,t\,,\,(c-1)\remainder n\,\big)
   & \text{if }\matr{\rho}(t,c)=1 \\
\big(\,\alpha+[(t-\alpha+1)\remainder\delta]\,,\,(c-2)\remainder n\,\big) & \text{if }\matr{\rho}(t,c)=0
\end{cases}$
\end{itemize}
where $x\remainder y$ is the remainder upon dividing $x$ by $y$, \ie,
the unique $z\in\{0,\ldots,y-1\}$ such that
$z=x+yk$ for some $k\in \Z$.
The left wrap trajectory $\leftwraptraj{\rho}$ of our running example
$\rho$ is shown top-right (Figure~A): 
$\leftwraptraj\rho(0)$ to $\leftwraptraj\rho(6)$ are, in order,
$(0,5)$, $(0,4)$, $(1,2)$, $(1,1)$, $(2,5)$, $(2,4)$, $(3,2)$.
Note the `wrap' from $\leftwraptraj{\rho}(27)=(11,3)$ (the bottom
row) to $\leftwraptraj{\rho}(28)=(3,2)$ (the row below the horizontal
rule marking the onset of periodicity $\alpha=3$).

The \defn{start of cyclicity} in the left wrap trajectory $\leftwraptraj\rho$ is the least
index $a$ such that $\leftwraptraj{\rho}(a)=\leftwraptraj{\rho}(a+d)$ for some $d>0$.  The least such $d$
is the \defn{period} of the trajectory, and $a+d$ is the \defn{index
of first recurrence}.
In the running example (Figure~A), $a=7$ and $d=21$
($\leftwraptraj{\rho}(7)=\leftwraptraj{\rho}(7+21)=(3,1)$).

Write $\lefttrajcomp\rho$ for the composite of $\rho$ and
$\leftwraptraj\rho$ defined by
$\lefttrajcomp\rho(i)\,=\,\rho(\leftwraptraj\rho(i))$, listing the symbols
along the trajectory $\leftwraptraj\rho$.  In our example (Figure~A), 
$\lefttrajcomp\rho(0)\lefttrajcomp\rho(1)\lefttrajcomp\rho(2)\cdots$ is (from left to right)	
$$1010101
1111000
1111000
1111000
1111000\cdots$$
obtained by simply reading the symbols encountered along the arrows.

For $c\in \{0,\ldots,n-1\}$ define the \defn{$n$-wrapped decrement} $c^-$
as $c-1$ if $c>0$ and $n$ if $c=0$.\footnote{\Ie, $c^-=(c-1)\remainder n$.}
For $(t,c)\in\N\times\{0,\ldots,n-1\}$ define $(t,c)^-=(t,c^-)$.
Define $\lefttrajcomp\rho^-$ by
$\lefttrajcomp\rho^-(i)=\rho\big(\leftwraptraj\rho(i)^-\big)$, listing the
(wrapped-)left neighbours of the symbols in the trajectory
$\leftwraptraj\rho$.  In our example (Figure~A), 
$\lefttrajcomp\rho^-(0)\lefttrajcomp\rho^-(1)\lefttrajcomp\rho^-(2)\cdots$ is 
$010000111101011\cdots$.

Let $X$ be a $\{0,1\}$-word.  We shall define
$\{0,\varzero,\varone,1\}$-words $\leftseed X$ and $\leftstem X$
called the \defn{left seed} and the \defn{left stem}, respectively.
Let $\rho$ be the $n$-wrapped run with $\rho_0=X$. Let $d$ be the
period of the left wrap trajectory $\leftwraptraj\rho$, let $a$ be its
start of cyclicity, and let $b=a+d$, the index of first recurrence.
Let $T$ be the restriction of $\leftwraptraj\rho$ to the domain
$\{0,\ldots,b-1\}$.  Define the $\{1,0\varzero,0\varone\}$-word $U$ of
length $b$ by
$$U(i)\;\;=\;\;\begin{cases}
1 & \text{if }\lefttrajcomp\rho(i)=1 \\
0\,\underline{\mbox{\Large$\star$}} & \text{if }\lefttrajcomp\rho(i)=0\text{, where }\mbox{\Large$\star$}=\lefttrajcomp\rho^-(i)
\end{cases}$$
In our running example (Figure~A), $U$ is 
$$1(0\varone)1(0\varzero)1(0\varzero)1\;\;\;
1111(0\varone)(0\varzero)(0\varone)\;\;\;
1111(0\varone)(0\varzero)(0\varone)\;\;\;
1111(0\varone)(0\varzero)(0\varone)$$
of length 28, where gaps highlight repetition.
Recall that $a$ is the start of cyclicity.
Define the left stem $\leftstem X\:=\:\rev{(U_{<a})^\concat}\;$ (the
reverse of the flattening of $U_{<a}$) and the left seed $\leftseed
X\:=\:\red{\rev{(U_{\ge a})^\concat}}\:$ (the reduction of the reverse
of the flattening of $U_{\ge a}$).
In our example, $a=7$, so $U_{<a}=1(0\varone)1(0\varzero)1(0\varzero)1$ and $U_{\ge a}=
\big[1111(0\varone)(0\varzero)(0\varone)\big]^3$, hence
$$\leftstem X\;\;=\;\; 1\varzero 0 1 \varzero 0 1 \varone 0 1
\hspace*{15ex}
\leftseed X\;\;=\;\; \varone 0\varzero 0\varone 0 1111$$

\paragraph*{The right wrap construction.}
The \defn{right wrap trajectory} in the matrix $\matr\rho$ is the
function (infinite sequence) $\rightseedtraj{\rho}:\N\to E$ defined by
the following recursion.\footnote{See
page~\pageref{right-wrap-informal} for a more informal presentation.}
Recall that $\delta$ is the period and $\alpha$ is the time of first
repetition.
\begin{itemize}
\item $\rightwraptraj{\rho}(0)=(0,0)$;\vspace*{-1ex}
\item if $\rightwraptraj{\rho}(i)=(t,c)$ then $\rightwraptraj{\rho}(i+1)=
\begin{cases}
\big(\,t\,,\,(c+1)\remainder n\,\big)
   & \text{if }\matr{\rho}(t,c)=1 \\
\big(\,\alpha+[(t-\alpha+1)\remainder\delta]\,,\,(c+2)\remainder n\,\big) & \text{if }\matr{\rho}(t,c)=0
\end{cases}$
\end{itemize}
The \defn{start of cyclicity}, \defn{period} and \defn{index of first
recurrence} are defined as for the left wrap trajectory.
Write $\righttrajcomp\rho$ for the composite of $\rho$ and
$\rightwraptraj\rho$ defined by
$\righttrajcomp\rho(i)\,=\,\rho(\rightwraptraj\rho(i))$, listing the symbols
along the trajectory $\rightwraptraj\rho$.
For $c\in \{0,\ldots,n-1\}$ define the \defn{wrapped increment} $c^+$ as $c+1$ if $c<n$ and $0$ if $c=n$.
For $(t,c)\in\N\times\{0,\ldots,n-1\}$ define $(t,c)^+=(t,c^+)$.
Define $\righttrajcomp\rho^+$ by
$\righttrajcomp\rho^+(i)=\rho\big(\rightwraptraj\rho(i)^+\big)$, listing the
(wrapped-)right neighbours of the symbols in the trajectory
$\rightwraptraj\rho$.

Let $Y$ be a $\{0,1\}$-word.  We shall define
$\{0,\varzero,\varone,1\}$-words $\rightseed Y$ and $\rightstem Y$
called the \defn{right seed} and the \defn{right stem}, respectively.
Let $\rho$ be the $n$-wrapped run with $\rho_0=Y$. Let $d$ be the
period of the right wrap trajectory $\rightwraptraj\rho$, let $a$ be its
start of cyclicity, and let $b=a+d$, the index of first recurrence.
Let $T$ be the restriction of $\rightwraptraj\rho$ to the domain
$\{0,\ldots,b-1\}$.  Define the $\{\varone,\varzero 0,\varzero
1\}$-word $U$ of length $b$ by
$$U(i)\;\;=\;\;\begin{cases}
\underline 1 & \text{if }\righttrajcomp\rho(i)=1 \\
\underline 0\,\mbox{\Large$\star$} & \text{if }\righttrajcomp\rho(i)=0\text{, where }\mbox{\Large$\star$}=\righttrajcomp\rho^+(i)
\end{cases}$$
Define the right stem by $\rightstem X\:=\:(U_{<a})^\concat\;$ (the flattening of
$U_{<a}$) and the right seed by $\rightseed X\:=\:\red{(U_{\ge a})^\concat}\:$
(the reduction of the flattening of $U_{\ge a}$).

\subsection{The Wolfram-Cook Turing machine $M$}

Let $\Sigma_M=\{0,\varzero,\dunno,\varone,1\}$.  A
\defn{configuration} $\langle\tau,h,q\rangle$ of the Wolfram-Cook
Turing machine $M$ is a $\Sigma_M$-state $\tau$ (\ie, a function
$\tau:\Z\to\{0,\varzero,\dunno,\varone,1\}$) called the \defn{tape}, a
\defn{head position} $h\in\Z$, and a \defn{head state}
$q\in\{\circ,\bullet\}$.
Write $\mathcal{C}$ for the set of all configurations.\footnote{Thus
$\mathcal{C}=\{0,\varzero,\dunno,\varone,1\}^\Z\times\Z\times\{\circ,\bullet\}$
where $X^Y$ denotes the set of all functions from $Y$ to $X$.}
Write $\langle
\tau,h,q\rangle\rightsquigarrow\langle\tau',h',q'\rangle$ if the following conditions are satisfied (formalising the transition 
table on page~\pageref{table}):
\begin{itemize}
\item $\tau'(c)=\tau(c)$ for all $c\neq h$\,;\footnote{The head writes to cell $h$ only; all other cells remain the same.}
\item if $\tau(h)\in\{0,1\}$ then $h'=h-1$, otherwise $h'=h+1$\,;\footnote{The head moves left from $0$ and $1$, and right otherwise.}
\item if $\tau(h)\in\{1,\varone\}$ then $q'=\bullet$, otherwise $q'=\circ$\,;\footnote{The next state is $\bullet$ (`carry') iff the head reads $1$ or $\varone$.}
\item if $q=\circ$ then the ordered pair $\langle \tau(h),\tau'(h)\rangle$ is in 
$\{\langle 0,\varzero\rangle,\langle
1,\varone\rangle,\langle\varzero,0\rangle,\langle\varone,1\rangle,\langle\dunno,0\rangle\}$\,;\footnote{Corresponding to the first column in the table on page~\pageref{table}.}
\item if $q=\bullet$ then the ordered pair $\langle\tau(h),\tau'(h)\rangle$ is in 
$\{\langle 0,\varone\rangle,\langle
1,\dunno\rangle,\langle\varzero,0\rangle,\langle\varone,1\rangle,\langle\dunno,1\rangle\}$\,.\footnote{Corresponding
to the second column in the table on page~\pageref{table}.}
\end{itemize}
$M$ is deterministic: for any configuration $C$ there is a unique configuration $C'$
such that $C\rightsquigarrow C'$.
A \defn{run} of
$M$ is a function $\mu:\N\to\mathcal{C}$ such that
 $\mu(t)\rightsquigarrow\mu(t+1)$ for all $t$.
By determinism of $M$, the run is uniquely determined by the \defn{initial configuration} $\mu(0)$.
The \defn{initial tape}, \defn{initial head position} and
\defn{initial head state} are the tape, head position and head state
of the initial configuration.
Let $\mu$ be a run of $M$, and define the functions $\tau_\mu$,
$h_\mu$ and $q_\mu$ by
$\mu(t)=\langle\tau_\mu(t),h_\mu(t),q_\mu(t)\rangle$, the first of
which is the \defn{tape function}.
For $t>0$, if $h_\mu(t+1)=h_\mu(t-1)=h_\mu(t)-1$ the head \defn{switches left at
time $t$}
and if $h_\mu(t+1)=h_\mu(t-1)=h_\mu(t)+1$ the head \defn{switches
right at $t$}.
The head \defn{switches left at time\/ $0$} if $h_\mu(1)=h_\mu(0)-1$ (the first head
move is to the left).
The configurations $\mu(t)$ of a run $\mu$ are depicted on page~\pageref{emulation},
for $t$ such that the head switches left at time $t$.
The next figure (page~\pageref{interleaved}) in addition shows
the configurations $\mu(t)$ for $t$ such that the head switches right
at time $t$.

Let the set $L$ comprise those $t\in \N$ such that the head of $\mu$
switches left at time $t$, and write $L(i)$ for the $i\nth$ largest
element of $L$ (its smallest element taken to be $L(0)$, the $0\nth$).
The \defn{emulation} $\emulation{\mu}$ produced by a run $\mu$ of $M$
is the result of restricting the tape to times
when the head switches left, and reindexing the timestamps to be
sequential from $0$.  Formally (recalling that
$\tau_\mu:\Z\to\{0,\varzero,\dunno,\varone,1\}$ denotes the tape
function of $\mu$),
$\emulation{\mu}(t)\,=\,\tau_\mu(L(t))\,:\,\Z\to\{0,\varzero,\dunno,\varone,1\}$.\footnote{Note
that $\emulation\mu(i)$ may be undefined for some $i$, for example if
the head ``runs off to infinity'' at some point.}
The lower figure on page~\pageref{emulation} shows the sequence
$\emulation{\mu}(0)$,
$\emulation{\mu}(1)$,\ldots,$\emulation{\mu}(12)$.

\subsection{The Causal Future Emulation Theorem}

For any $\{0,1\}$-word $I$ of length $l$, a \defn{run from $I$} of the
rule 110 cellular automaton $R$ is a run $\rho$ whose initial
configuration has $I$ in cells $0,\ldots,l-1$, \ie,
$\rho_0(c)=I(c)$ for $c\in\{0,\ldots,l-1\}$.
The \defn{causal future} $\future{I}$ of a run from $I$ is the
set of pairs $\{\langle t,c\rangle\,:\, -t\le c< k+t\}\subseteq 
\N\times\Z$.\footnote{Note that this depends only on the integer $l$.}

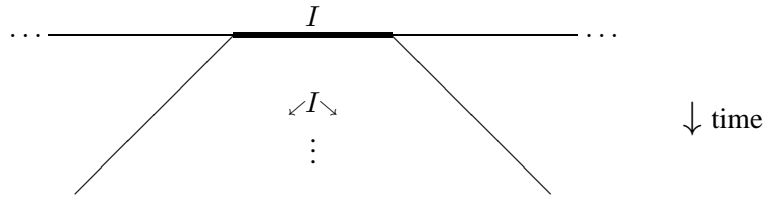
\begin{center}
\begin{picture}(0,70)(0,-60)
\put(-100,0){\line(1,0){200}}
\put(103,-.5){$\ldots$}
\put(-115,-.5){$\ldots$}
\put(-30,0){\line(-1,-1){60}}
\put(30,0){\line(1,-1){60}}
\thicklines\put(-30,0.7){\line(1,0){60}}
\thicklines\put(-30,-0.1){\line(1,0){60}}
\thicklines\put(-30,-0.5){\line(1,0){60}}
\put(0,-35){\makebox(0,0){\shortstack{$\future{\codedtaginput}$ \\ $\vdots$}}}
\put(0,3.5){\makebox(0,0)[b]{$\codedtaginput$}}
\put(140,-35){\Large$\downarrow$ \normalsize time}
\end{picture}
\end{center}
In Wolfram and Cook's proof of the universality of $R$, the state of
the emulated universal cyclic tag system on (transformed) input $I$ is
recovered from the future $\future{\codedtaginput}$ alone.
We shall refer to this property as \defn{future sufficiency}.

Analogously, for the Turing machine $M$, a \defn{run of from $I$} is a
run $\mu$ of $M$ whose initial tape has $I$ in cells $0$ to $l-1$, in
other words, $\tau_\mu(0,c)=I(c)$ for $c\in\{0,\ldots,l-1\}$.

The following theorem allows us to emulate the causal future $\future
I$ in the rule 110 automaton $R$ on the Turing machine $M$.
\begin{theorem}[Causal Future Emulation]
Let $A,I,B$ be $\{0,1\}$-words of lengths $a,l,b$, respectively, with
$A$ and $B$ each containing $0$.
Let $\rho$ be the run of the
rule 110 cellular automaton $R$ from $I$ with initial state
$$\cells \cell{\repleft A}\; \cell I\; \cell{\repright B} \cellsend$$
Let $\mu$ be the run of the Wolfram-Cook Turing machine $M$ from $I$
with initial tape\footnote{Recall that $\leftseed A$, $\leftstem A$,
$\rightstem B$ and $\rightseed B$ were defined by the wrap
constructions.}
$$\repleft {\leftseed A}
\;\;
\leftstem{A}
\:\;I\;\,\rightstem B\,\;\repright {\rightseed B}$$
and head in initial state $\circ$ and initial position $l$ (the cell immediately to the right of $I$).
Recall that $\emulation\mu(i)$ denotes the tape function
$\tau_\mu(t):\Z\to\{0,\varzero,\dunno,\varone,1\}$ for the time $t$ of
the $i\nth$ occasion the head switches left.
Then
for all events $\langle t,c\rangle$ in the causal future $\future I\,$,
$$\rho(t,c)\;\;=\;\;\big(\emulation\tau(t)\big)(c)\;.$$
\end{theorem}
The proof of the theorem is the subject of the next section.
Universality is an immediate corollary:
\begin{theorem}[Universality]
Wolfram and Cook's 2-state 5-symbol Turing machine is universal.
\end{theorem}
\begin{proof}
Take $A$ and $B$ in the Causal Future Emulation Theorem to be the
infinitely repeated words used by Wolfram and Cook to simulate a
cyclic tag system in the rule 110 automaton $R$ \cite{C,W}.
Universality follows from future sufficiency.
\end{proof}

\section{Proof of the Causal Future Emulation Theorem}

Before working through this proof, the reader may wish to study the
figure on page~\pageref{leftright} which shows a causal future
emulation at both left- and right switches of the head.

[\ldots]

\subsection*{Acknowledgement}
Many thanks to Rob van Glabbeek for detailed feedback.  Thanks also to
Vaughan Pratt.

\end{document}